\documentclass[a4paper,UKenglish,cleveref, autoref, thm-restate,authorcolumns]{lipics-v2019}
\title{Symbolic Execution Game Semantics}

\author{Yu-Yang Lin}{Queen Mary University of London, UK}{}{}{}%TODO mandatory, please use full name; only 1 author per \author macro; first two parameters are mandatory, other parameters can be empty. Please provide at least the name of the affiliation and the country. The full address is optional

\author{Nikos Tzevelekos}{Queen Mary University of London, UK}{}{}{}

\authorrunning{Lin et al.} %TODO mandatory. First: Use abbreviated first/middle names. Second (only in severe cases): Use first author plus 'et al.'

\Copyright{Yu-Yang Lin, Nikos Tzevelekos} %TODO mandatory, please use full first names. LIPIcs license is "CC-BY";  http://creativecommons.org/licenses/by/3.0/

 %TODO mandatory: Please choose ACM 2012 classifications from https://dl.acm.org/ccs/ccs_flat.cfm 

 %TODO mandatory; please add comma-separated list of keywords

\EventEditors{John Q. Open and Joan R. Access}
\EventNoEds{2}
\EventLongTitle{42nd Conference on Very Important Topics (CVIT 2016)}
\EventShortTitle{CVIT 2016}
\EventAcronym{CVIT}
\EventYear{2016}
\EventDate{December 24--27, 2016}
\EventLocation{Little Whinging, United Kingdom}
\EventLogo{}
\SeriesVolume{42}
\ArticleNo{23}
\usepackage{comment,paralist}
\usepackage{url,hyperref}
\usepackage{wrapfig}
\usepackage{latexsym,amssymb,amsmath,ae,amscd,stmaryrd}
\usepackage{graphicx}% http://ctan.org/pkg/graphicx, for scalebox
\usepackage{enumerate}% http://ctan.org/pkg/enumerate
\usepackage{listings} % for code
\usepackage{mathtools}
\usepackage{array} % for \newcolumntype macro
\usepackage{ebproof,proof} % for proof trees
\usepackage{amsmath}
\usepackage{xtab}
\usepackage{xspace} 
\usepackage{enumitem}

\usepackage{csvsimple}
\usepackage{pdfpages}

\usepackage{amsthm} 
\defaultleftmargin{1.5em}{}{}{}
\usepackage{xcolor}

\newcommand\nin{\not\in}
\newcommand\cutout[1]{}

\newcommand\sints{\mathtt{SInts}}
\newcommand\Tunit{\atype{unit}}
\newcommand\Tint{\atype{int}}
\newcommand{\refR}{\mathtt{global}~}
\newcommand\comp{\text{\bf ;}}
\newcommand\textilde[1]{$\widetilde{\text{#1}}$}

\usepackage{microtype}

\newcommand\shorten[2]{#2}

\newcommand{\ret}{ret} %  ret
\newcommand{\assert}{\mathtt{assert}} %  assert
\newcommand{\refs}{\mathtt{Refs}}
\newcommand{\meths}{\mathtt{Meths}}
\newcommand{\vars}{\mathtt{Vars}}
\newcommand{\ints}{\mathtt{Int}}

\newcommand{\globalL}{\mathtt{global}~}
\newcommand{\publicL}{\mathtt{public}~}
\newcommand{\abstractL}{\mathtt{abstract}~}

\newcommand\atype[1]{\mathtt{#1}}

\newcommand\symto{\to_s}
\newcommand\xsymto[1]{\xrightarrow{#1}_s}

\newcommand\compsemto{\to_{1,2}}

\newcommand\xgamto[1]{\xrightarrow{#1}_G}
\newcommand\ite[3]{\mathtt{if~}#1\mathtt{~then~}#2\mathtt{~else~}#3}
\newcommand\letm[2]{\mathtt{let~}#1\mathtt{~in~}#2}
\newcommand\letrec[2]{\mathtt{letrec~}#1\mathtt{~in~}#2}
\newcommand\ifm[3]{\mathtt{if}~#1~\mathtt{then}~#2~\mathtt{else}~#3}

\newcommand{\xtwoheadrightarrow}[2][]{%
  \xrightarrow[#1]{#2}\mathrel{\mkern-14mu}\rightarrow
}

\newcommand{\holik}{\text{HOLiK}}
\newcommand{\holib}{\text{HOLi}}

\newcommand{\coneqct}{\text{Coneqct}}
\newcommand{\mochi}{\text{MoCHi}}
\newcommand{\klee}{\text{KLEE}}

\newcommand{\jcute}{\text{jCUTE}}
\newcommand{\cbmc}{\text{CBMC}}
\newcommand{\reguard}{\text{ReGuard}}
\newcommand{\oyente}{\text{Oyente}}
\newcommand{\majan}{\text{Majan}}

\newcommand{\scv}{\text{SCV}}

\newcommand{\racket}{\text{Racket}}

\newcommand{\kfw}{\mathbb{K}}

\newcommand\boldemph[1]{\textbf{\emph{#1}}}
\newcommand\dom{\mathrm{dom}}

\newcommand\sem[1]{\llbracket #1\rrbracket}
\newcommand{\pair}[1]{\langle #1 \rangle}

\newcommand{\eval}{\mathcal{E}}
\newcommand{\pub}{\mathcal{P}}
\newcommand{\abs}{\mathcal{A}}
\DeclarePairedDelimiter\sizeof{\lvert}{\rvert}%

\newcommand\evalbox[1]{\llparenthesis{#1}\rrparenthesis}
\newcommand\evaltag[2]{\evalbox{#1}^{\pair{#2}}}

\newcommand\preG{} %{\Gamma\vdash}
\newcommand{\syncomp}{\check{;}}
\newcommand{\confcomp}{\curlywedge}
\newcommand{\semcomp}{\oslash}

\newcolumntype{L}{>{$}l<{$}} % math-mode version of "l" column type
\newcolumntype{R}{>{$}r<{$}} % math-mode version of "r" column type
\newcolumntype{C}{>{$}c<{$}} % math-mode version of "c" column type

\makeatletter
\DeclareRobustCommand\bigop[1]{%
	\mathop{\vphantom{\sum}\mathpalette\bigop@{#1}}\slimits@
}
\newcommand{\bigop@}[2]{%
	\vcenter{%
		\sbox\z@{$#1\sum$}%
		\hbox{\resizebox{\ifx#1\displaystyle1.8\fi\dimexpr\ht\z@+\dp\z@}{!}{$\m@th#2$}}%
	}%
}
\makeatother

\usepackage{graphics} %needed for the NEW quantifier
 
\allowdisplaybreaks

\begin{document}
\maketitle

\begin{abstract}
We present a framework for symbolically executing and model checking higher-order programs with external (open) methods.
We focus on the client-library paradigm and in particular we aim to check libraries with respect to any definable client.
We combine traditional symbolic execution techniques with operational game semantics to build a symbolic execution semantics that captures arbitrary external behaviour.
We prove the symbolic semantics to be sound and complete.
This yields a bounded technique by imposing bounds on the depth of recursion and callbacks.
We provide an implementation of our technique in the $\kfw$ framework and showcase its performance on a custom benchmark based on higher-order coding errors such as reentrancy bugs.
\end{abstract}

\section{Introduction}
{Two important challenges in program verification} are {state-space explosion}  and the {environment problem}.
The former refers to the need to investigate infeasibly many states, while the latter concerns cases where the code depends on an environment that is not available for analysis.
% Techniques that bound the depth of the analysis, like bounded model checking (BMC)~\cite{DBLP:conf/cav/BiereCRZ99}, help mitigate the state-space problem.
%
% {Such techniques} have proved extremely successful and are used nowadays to catch bugs in software at an industrial scale~\cite{DBLP:conf/tacas/ClarkeKL04,DBLP:conf/tacas/MorseRCN014,survey:verification}.
%
State-space explosion has been approached with a range of techniques, which have led to verification tools being nowadays routinely used on industrial-scale code (e.g.~\cite{ASTREE,INFER,CBMC}).
The environment problem, however, remains largely unanswered: {verification techniques often require} the whole code to be present for the analysis and, in particular, cannot analyse components like libraries where parts of the code are missing (e.g.\ the client using the library).
This problem is particularly acute in higher-order programs, where the interaction between a program and its environment can be intricate and e.g.\ involve callbacks or reentrant calls.
In this paper we address this latter problem by combining \emph{game semantics}, a semantics theory for higher-order programs, with \emph{symbolic execution}, a technique that uses \emph{symbolic values} to explore multiple execution paths of a program.

To showcase the importance and challenges of the environment problem, following is a simple example of a library written in a sugared version of \holib, the vehicle language of this paper. The example is a simplified implementation of ``The DAO'' smart contract, a failed decentralised autonomous organisation on the Ethereum blockchain platform~\cite{dao}. As with
\begin{wrapfigure}[9]{l}{.51\linewidth}\vspace{-4mm} %used to be -4mm
\begin{lstlisting}[label=dao]
  import send:(int -> unit)
  int balance := 100;

  public withdraw (m:int) :(unit) = 
    if (not (!balance < m)) then 
         send(m);
         balance := !balance - m;
         assert(not(!balance < 0))
    else ();
\end{lstlisting}\end{wrapfigure}
libraries, the challenge in analysing smart contracts is that the client code is not available. We must thus generate all possible contexts in which the contract can be called. In this case, the error is caused by a reentrant call from the \texttt{send()} method, which is provided by the environment. When this method is called, the environment takes control and is allowed to call any method in the library. If a client were to call \texttt{withdraw()} within its \texttt{send()} method, the recursive call would drain all the funds available, which is simulated in this example by a negative balance. This happens because the method is manipulating a global state, and is updating it after the external call.
We can see that an analysis capturing this error would need to be able to predict an intricate environment behaviour. Moreover, such an analysis should ideally  only predict realisable environment behaviours.

Symbolic execution~\cite{symbolic:select,symbolic:dissect,symbolic:king}  explores all paths of a program using symbolic values instead of concrete input values. Each symbolic path holds a path condition (a SAT formula) that is satisfiable if and only if the path can be concretely executed.
While the resulting analysis is unbounded in general, by restricting our focus to bounded paths we can soundly catch errors, or affirm the absence thereof up to the used bound.
Game semantics~\cite{GS,GS2}, on the other hand, models higher-order program phrases in isolation as 2-player games: sequences of computational \emph{moves} (method calls and returns) between the program and its hypothetical environment. The power of the technique lies in its use of combinatorial conditions to precisely allow those game plays that can be realised by including the program in an actual environment. Moreover, the theory can be formulated operationally in terms of a trace semantics for open terms~\cite{OGS1,OGS2,OGS3} which, in turn,
lends itself to a symbolic representation.
The latter yields a symbolic execution technique that is \emph{sound and complete} in the following sense: given an open program, its symbolic traces match its concrete traces, which match its realisable traces in some environment.

Returning to the DAO example,
we can model the ensuing interaction as a sequence of moves, alternating between the environment and the library. Any finite sequence of moves (that leads to an assertion violation) is a trace defining a counterexample. 
Running the example in {\holik}, our implementation of the symbolic semantics in the $\kfw$ Framework~\cite{k:framework}, the following minimal symbolic trace is automatically found:
\begin{align*}\\[-6mm]
&call\pair{withdraw,x_1} \cdot call\pair{send,x_1} \cdot call\pair{withdraw,x_2}\\
&\qquad\cdot call\pair{send,x_2} \cdot ret\pair{send,()} \cdot ret\pair{withdraw,()} \cdot ret\pair{send,()}\\[-5.75mm]
\end{align*}
where $x_1$ is the original call parameter, and $x_2$ is the parameter for the reentrant call, satisfiable with values $x_1 = 100$ and $x_2 = 1$. A fix would be to swap line 6 and 7, to update internal state before passing control.
% The actual DAO attack cost over \$100m and resulted in a fork that split the Ethereum platform.

In Appendix~\ref{app:examples} we look at a few more examples of libraries that exhibit errors due to high-order behaviours. We provide three examples: a file lock example, a double deallocation example, and an unsafe implementation of flat-combining.

Overall, this paper contributes a novel symbolic execution technique based on game semantics to precisely model the behaviour of higher-order stateful programs. Specifically:
\begin{inparaitem}
\item We present a symbolic trace semantics for higher-order libraries that captures the behaviour of an unknown environment, and prove it sound and complete: i.e.\ it produces no spurious error traces, and is able to produce the complete execution tree of any library.
\item By bounding the depth of nested calls and the \emph{insistence} of the environment in calling library methods, we derive a sound and bounded-complete technique to check higher-order libraries for errors. 
\item We implement the latter in the $\kfw{}$ semantical framework~\cite{k:framework} to produce a sound and bounded-complete tool for higher-order libraries as a proof of concept. We test our implementation with benchmarks adapted from the literature.
\end{inparaitem}
Some material has been delegated to an Appendix. \shorten{Full proofs can be found in an extended version of the paper~\cite{FULL}.}{}

%%% Local Variables:
%%% mode: latex
%%% TeX-master: "paper"
%%% End:

\section{A Language for Higher-Order Libraries: \holib}

We introduce \holib, a language for higher-order libraries which define methods to be used by an external client, and in turn require external methods (provided by the client).
We give in \holib\ an operational semantics for terms that {integrates a counter for the depth of nested calls that a program phrase can make}.
We then extend {this counting semantics} to open terms by means of a trace semantics. We show that the trace semantics of libraries is sound and complete for reachability of errors under any external client.

%The semantics is accompanied by a symbolic version thereof, which is shown to be sound: a value or error is reached symbolically iff it can be reached concretely as well.

\subsection{Syntax and operational semantics}

\newcommand\qweqwe{\ \ \,}

\begin{figure}[t]\small
\[
\begin{array}{@{\!\!\!\!\!\!\!\!\!\!}l@{\!\!\!\!\!\!\!\!\!\!\!\!\!\!\!\!\!\!\!\!}l}
\begin{array}{@{}rrl}
  & \mathit{Libraries} \ \ L ::=& B \mid \abstractL m; L
  \\[.75mm]
  & \mathit{Blocks}\ \ B ::=& \varepsilon \mid \publicL m = \lambda x. M; B\\[.5mm]
&&\mid  m = \lambda x. M; B
   \mid \refR r := i; B\\[.5mm]
  && \mid \refR r := \lambda x.M; B
\\[.75mm]
& \mathit{Clients} \ \ C \ ::=&\ L;\, \mathtt{main} = M
\end{array}
  &
    \begin{array}{rrl}
	& \mathit{Terms} \  \ M ::=& m \mid i \mid ()\mid x\mid \lambda x.M
	\mid r:=M 
	\mid {!r}\\[.5mm] 
	&& \mid M \oplus M 
	\mid \langle M,M \rangle 
	\mid \pi_1 M 
	\mid \pi_2 M \\[.5mm] &&
	\mid MM
	\mid \ite{M}{M}{M} \\[.5mm]
	&&
           \mid  \letrec{x = \lambda x.M}{M}\\[.5mm]
      && \mid \letm{x=M}{M}
           \mid  \assert(M)
%  \\[1mm]
%	& \quad \vals \ \ni \ v ::=&m \mid i \mid () \mid \pair{v,v}\\[-2mm]
    \end{array}\\[-1mm]
\end{array}\]
\hrule
\begin{gather*}\\[-8mm]
%=====Unit=====
{\infer{ \preG () : \atype{unit} }{}}
\qweqwe
%=====Int=====
{\infer{ \preG i : \atype{int} }{}}
\qweqwe
%=====Variables=====
\infer{ \preG x : \theta }{ x\in\vars_\theta }
\qweqwe
%=====Methods=====
\infer{ \preG m : \theta \rightarrow \theta' }{ m \in \meths_{\theta,\theta'} }
\qweqwe
%=====Arithmetic=====
\infer{  M \oplus M' : \atype{int} }{  M, M' : \atype{int} }
\qweqwe		
%=====Conditionals=====
\infer{ \preG \ifm{M}{M_1}{M_0} : \theta }{ \preG M : \atype{int}\qquad  \preG M_1,M_0 : \theta }
\\[1mm]
%=====Pairs=====
\infer{ \preG \langle M , M' \rangle : \theta_1 \times \theta_2 }
  { \preG M : \theta_1 \quad
	\preG M' : \theta_2 }
\qweqwe
%=====Projection=====
\infer{ \preG \pi_i \langle M , M' \rangle : \theta_i}{ \preG \langle M , M' \rangle : \theta_1 \times \theta_2 }
\qweqwe
%=====Dereferencing=====
\infer{ \preG {!r} : \theta }{ r \in \refs_{\theta} }
\qweqwe
% =====Assigning Ints=====
\infer{\preG r := M : \atype{unit} }{ r \in \refs_{\theta} \quad \preG M : \theta }
\qweqwe 
%=====Function Application Vars=====
\infer{ \preG M'\,M : \theta' }
{ M': \theta \rightarrow \theta' \quad \preG M : \theta }
\\[1mm]
%=====Lambdas=====
\infer{ \preG \lambda x . M : \theta \rightarrow \theta' }{ M : \theta'\quad  x : \theta }
\qweqwe
%=====Let bindings=====
\infer{ \preG \letm{x = M}{M'} : \theta' }{ \preG x,M : \theta \quad  M' : \theta' }
\qweqwe 
%=====Letrec bindings=====
\infer{ \preG \letrec{x = \lambda y . M}{M'} : \theta'}{ x,\lambda y.M : \theta \to \theta'' \quad
M' : \theta' }
\qweqwe
% =====Fail=====
\infer{ \preG \assert(M) : \atype{unit} }{ M : \atype{int} }\\[-9mm]
\end{gather*}
\caption{Syntax and typing rules of \holib.}\label{fig:holib}
\end{figure}

A library in  \holib\ is a collection of typed higher-order methods. A client is simply a library with a main body. Types are given by the grammar:
\[\theta\ ::= \ \atype{unit}\mid \atype{int} \mid \theta\times\theta\mid \theta\to\theta\]

We use countably infinite sets $\meths$, $\refs$ and $\vars$ for method, global reference and variable names, ranged over by $m$, $r$ and $x$ respectively, and variants thereof; while
$i$ is for ranging over the integers.
We use $\oplus$ to range over a set of binary integer operations, which we leave unspecified.
Each set of names is typed, that is, it can be expressed as a disjoint union as follows: \
$
\meths = \biguplus\nolimits_{\theta,\theta'} \meths_{\theta,\theta'}, \;\;
\refs = \biguplus\nolimits_{\theta\not=\theta_1\times\theta_2} \refs_\theta,\;\;
\vars = \biguplus\nolimits_\theta \vars_\theta.
$

The full syntax and typing rules are given in Figure~\ref{fig:holib}.
%
%
%For instance, a method $run\in\meths_{\Tunit,\Tunit}$, is of type $\Tunit\to\Tunit$.
%
Thus, a library consists of
abstract method declarations, followed by blocks of public and private method and reference definitions. 
A method is considered private unless it is declared $\mathtt{public}$.
Each public/private method and reference is defined once.
Abstract methods are not given definitions: these methods are external to the library. Public, private and abstract methods are all disjoint.

Libraries are well typed if all their method and reference definitions are well typed (e.g.\ $\publicL m=\lambda x.M$ is well typed if $m:\theta$ and $\lambda x.M:\theta$ are both valid for the same type $\theta$) and only mention methods and references that are defined or abstract.
A client $L;\, \mathtt{main} = M$ is well typed if
$M:\Tunit$ is valid and
$L;\, m = \lambda x.M$ is well typed for some fresh $x,m$.
A library/client is \emph{open} if it contains abstract methods. This is different to open/closed terms: we call a term \emph{open} if it contains free variables.

\begin{remark}
By typing variable, reference and method names, we do not need to provide a context in typing judgements. 
Note that the references we use are of non-product type and, more importantly, \textbf{global} to the library: a term can use references but not create them locally or pass them as arguments (we discuss how to include such references in Appendix~\ref{general:refs}). 
\end{remark}

\begin{example}
The DAO-attack example from the Introduction can be written in \holib\ as:
\begin{align*}
  & \abstractL send;\;\;
   \globalL bal := 100;\\
  & \publicL wdraw =\\
    & \qquad\quad \lambda x.\ \ifm{{!bal} \geq x} 
   {({send}(x); bal := {!bal}-x; \assert({!bal}\geq0))}{()}
\end{align*}
where $send,wdraw \in\meths_{\Tint,\Tunit}$, $bal\in\refs_{\Tint}$, and $M;M'$ stands for $\letm{\_=M}{M'}$.
\end{example}

% The language presented here is a variation of the calculus presented in \cite{DBLP:journals/corr/MurawskiT16a}.
%
%We use $\meths(L)$ and $\refs(L)$ for the set of all methods and all references in $L$ respectively. We use similar notation for terms and blocks.

%\subsection{Operational Semantics}

\begin{figure}[t]\small
\[\begin{array}{ll}
(E[\letm{x=v}{M}],R,S,k)\to(E[M\{v/x\}],R,S,k) &
(E[\pi_j\pair{v_1,v_2}],R,S,k)\to(E[v_j],R,S,k)\\[1mm]
(E[r:=v],R,S,k)\to(E[()],R,S[r \mapsto v],k) &
(E[!r],R,S,k)\to(E[S(r)],R,S,k) \\[1mm]
(E[\ifm{i}{M_1}{M_0}],R,S,k)\to(E[M_j],R,S,k)\;\; (1) 
& (E[i_1 \oplus i_2],R,S,k)\to(E[i],R,S,k)\;\;(2) \\[1mm]
    (E[\lambda x.M],R,S,k)\to(E[m],R\uplus\{m\mapsto\lambda x.M\},S,k)
& (E[\assert(i)],R,S,k)\to(E[()],R,S,k)\;\; (3) 
\\[1mm]
%&(E[\ifm{i}{M_1}{M_0}],R,S,k)\to(E[M_1],R,S,k)\quad(i\neq 0)\\
(E[mv],R,S,k)\to(E[\evalbox{M\{v/x\}}],R,S,k+1)\;\; (4) &
    (E[\evalbox{v}],R,S,k+1)\to(E[v],R,S,k)\\[1mm]
\multicolumn{2}{l}{(E[\letrec{f=\lambda x.M}{M'}],R,S,k)\to(E[M'\{m/f\}],R\uplus\{m \mapsto \lambda x.M\{m/f\}\},S,k)}\\[1mm]
    \multicolumn{2}{l}{\text{Conditions: }
    (1): j=1\text{ iff }i\not=0, \quad
    (2): i=i_1\oplus i_2,\quad (3): i\neq 0,\quad
(4): R(m)=\lambda x.M.
    }\\[-2mm]
\end{array}\]\hrule
\begin{flalign*}\\[-7mm]
  \mathit{Values} \ \  v \ &::=\ m \mid i \mid () \mid \pair{v,v}
  \qquad\qquad\quad
  \mathit{Terms\ (extended)} \ \ M ::= \ \dots \mid \evalbox{M}\\
\mathit{Eval.\, Contexts} \ \ E &::= ~~{\bullet} \mid \assert(E) \mid r:=E \mid E \oplus M \mid v \oplus E \mid \pair{E,M} \mid \pair{v,E} \mid \pi_j E \\
&\qquad\mid mE\mid \letm{x=E}{M} \mid \ifm{E}{M}{M} \mid \evalbox{E} \\[-6.5mm]
\end{flalign*}
\hrule
\begin{flalign*}
  \\[-8mm]
%&(L)\xrightarrow{bld}(L,\varnothing,\varnothing) \qquad
(\abstractL m; L,R,S,\pub,\abs) &\xrightarrow{bld}(L,R,S,\pub,\abs\uplus\{m\}) \\[-1mm]
(\publicL m=\lambda x.M;B,R,S,\pub,\abs) &\xrightarrow{bld} (B,R\uplus\{m \mapsto \lambda x.M\},S,\pub\uplus\{m\},\abs)\\[-1mm]
(m=\lambda x.M;B,R,S,\pub,\abs) &\xrightarrow{bld} (B,R\uplus\{m \mapsto \lambda x.M\},S,\pub,\abs)\\[-1mm]
 (\refR r:=i;B,R,S,\pub,\abs) &\xrightarrow{bld} (B,R,S\uplus\{r\mapsto i\},\pub,\abs)\\[-1mm]
(\refR r:=\lambda x.M;B,R,S,\pub,\abs) &\xrightarrow{bld} (B,R\uplus\{m \mapsto \lambda x.M\},S\uplus\{r \mapsto m\},\pub,\abs) \\[-8mm]
\end{flalign*}
\caption{Operational semantics (top); values and evaluation contexts (mid); library build (bottom).\!\!}\label{fig:libs}\label{fig:opsems}
\end{figure}

A library contains public methods that can be called by a client. On the other hand, a client contains a main body that can be executed. These two scenarios constitute the operational semantics of \holib. Both are based on evaluating (closed) terms, which we define next.
Term evaluation requires: the closed term being evaluated;
method definitions, provided by a method repository;
reference values, provided by a store;
and a {call-depth counter} (a natural number).
Since method application is the only source of infinite behaviour in \holib, bounding the depth of nested calls is enough to guarantee {termination in program analysis. Hence we provide a mechanism to keep track of call depth.}

The operational semantics is given in Figure~\ref{fig:opsems}. The evaluation of terms (top part) involves configurations of the form
$(M,R,S,k)$, where:
\begin{itemize}
\item $M$ is a closed term 
which may contain \emph{evaluation boxes}, i.e.\ points inside a term where a method call has been made and has not yet returned, and is taken from the syntax 
extending the one of Figure~\ref{fig:holib} with the rule: \
$M ::= \ \dots \mid \evalbox{M}$
\item $R$ is a \emph{method repository}, i.e.\ a partial map from method names to their bodies
\item $S$ is a \emph{store}, i.e.\ a partial map from reference names to their stored values
\item $k$ is a \emph{counter}, i.e.\ a natural number.
\end{itemize}
Most of the rules are standard, but
it is worth noting that lambdas are not values themselves but, rather, evaluate to method names that are freshly stored in the repository. Moreover,
evaluation boxes interplay with the {counter} $k$ in the semantics:
they mark places where the {depth} has increased because of a nested call.
%
%The  operational semantics of closed terms is given in Figure~\ref{fig:opsems} (bottom). 
The penultimate line of rules in the operational semantics keeps track of call depth, and illustrates the utility of evaluation boxes: making a call {increases the counter} and leaves behind an evaluation box; returning form the call removes the box and {decreases the counter} again.

%As a result, in a given configuration $(M,R,S,k)$, the bound $k$ is the initial bound determines the current nesting depth.

%We next provide a bounded operational semantics for $\holib$.
A library $L$ \emph{builds} into a configuration of the form $(\varepsilon,R,S,\pub,\abs)$, which includes its public methods according to the rules in Figure~\ref{fig:libs} (bottom).
More precisely, $R$ and $S$ are as above, while  $\pub,\abs\subseteq\meths$ are (disjoint) sets of \emph{public} and \emph{abstract} method names. 
We say that (a well typed) $L$ builds to $(\varepsilon,R,S,\pub,\abs)$ if  $(L,\emptyset,\emptyset,\emptyset,\emptyset)\xrightarrow{bld}^*(\varepsilon,R,S,\pub,\abs)$.
If $L$ builds to $(\varepsilon,R,S,\pub,\abs)$
then the client $L;\, \mathtt{main} = M$ builds to $(M,R,S,\pub,\abs)$.
Moreover, we can link libraries to clients and evaluate them, as in the following definition.

\begin{definition}
\begin{enumerate}
\item
Library $L$ and client $C$ are \emph{compatible} if
$L$ builds to some $(\varepsilon,R,S,\pub,\abs)$ and $C$ builds to some 
$(M,R',S',\pub',\abs')$ such that:
$\pub=\abs'$ and $\abs=\pub'$ (\emph{complementation});  
  $\dom(S)\cap\dom(S')=\emptyset$ (\emph{disjoint state}); and 
  $\dom(R)\cap\dom(R')=\emptyset$ (\emph{method ownership}).
% \begin{itemize}
% \item $\pub=\abs'$ and $\abs=\pub'$ (\emph{complementation});  
% \item $\dom(S)\cap\dom(S')=\emptyset$ (\emph{disjoint state}); and 
% \item $\dom(R)\cap\dom(R')=\emptyset$ (\emph{unique method ownership}).
% \end{itemize}
  \item
For a library $L$, we let 
$\hat L$ be $L$ with all its abstract method declarations and $\mathtt{public}$ keywords removed; and similarly for $\hat C$.
Given compatible library $L$ and client $C$, we let their \emph{composition} be the client:
$L\comp C = \hat L ; \hat C$.
\item   Given compatible $L,C$, the 
  semantics of $L\comp C$ is:
  \[
    \sem{L\comp C} = \{ \rho \mid L\comp C\text{ builds to }(M,R,S,\emptyset,\emptyset)\land
    (M,R,S,{0})\to^*\rho\}
    \]
We  say that $\sem{L\comp C}$ \emph{fails} if it contains some $(E[\assert(0)],\cdots)$.
\end{enumerate}
\end{definition}

\begin{example}
To illustrate how libraries and clients are used, consider the DAO example again as a library $L_\mathtt{DAO}$. We can define a client
$C_\mathtt{atk}$:
\begin{align*}
  & \abstractL wdraw; \; \globalL wlet := 0;\\
  & \publicL send = \lambda x. wlet := {!wlet} + x; 
  \;  \ifm{!wlet < 100}{wdraw(x)}{()}  ; \\
  & \mathtt{main} = wdraw(1)
\end{align*}
to produce the following linked client ${L}_\mathtt{DAO} \comp {C}_\mathtt{atk}$ (modulo re-ordering):
\begin{align*}
  & \globalL bal := 100;\;   \globalL wlet := 0; \\ 
  & wdraw = \lambda x.\ \ifm{{!bal} \geq x}
   {({send}(x);
bal := {!bal}-x; \assert({!bal}>0))}
{()} ;\\
  & \publicL send = \lambda x. wlet := {!wlet} + x; 
  \;  \ifm{!wlet < 100}{wdraw(x)}{()}  ; \\
  &\mathtt{main} = wdraw(1)
\end{align*}
We can see how $L_\mathtt{DAO}$ is vulnerable to an attacker such as $C_\mathtt{atk}$ after linking them. The aim is thus to use bounded analysis to find counterexamples that define clients such as this one.
\end{example}

\subsection{Trace Semantics}\label{sec:traces}

The semantics we defined only allows us to evaluate terms, and only so long as their method applications only involve methods that can be found in the repository $R$. We next extend this semantics to encompass libraries and terms that can also call abstract methods.
The approach we follow is based on operational game semantics~\cite{OGS1,OGS2,OGS3} and in particular the semantics is given by means of traces of method calls and returns (called \emph{moves} in game semantics jargon), between the library and its client. In between such moves, the semantics evolves as the operational semantics we already saw.

To maintain a {terminating analysis}, we need to {keep track of} an added source of infinite execution, namely endless {consecutive calls from an external component}: a library will never terminate if its
client keeps calling its methods. This leads us to a semantics with two {counters}, $k$ and $l$, where $k$ {keeps track of} internal nested method calls and $l$ {records the number of consecutive calls made from the external component}. This counter $l$ is orthogonal to $k$ and is refreshed at every call to the external context.

When computing the semantics of a library, the library and its methods are the \emph{Player (P)} of the computation game, while the (intended) client is the \emph{Opponent (O)}. As the semantics is given in absence of an actual client, $O$ actually represents every possible client. When computing the semantics of a client, the roles are reversed. In both cases, though, the same sets of rules is used and there is no need to specify who is $P$ and $O$ in the semantics.

\begin{figure}[t]\small
\begin{tabular}{LL}
\begin{tabular}{l}\!\!\!(\textsc{INT})\\[1mm]\text{} \end{tabular} &
\infer{(\eval,M,R,S,\pub,\abs,k)_p
                \rightarrow (\eval,M',R',S',\pub,\abs,k')_p}
                {(M,R,S,k)\to(M',R',S',k')}
\\[-1mm]
(\textsc{PQ})&
(\eval,E[mv],R,S,\pub,\abs,k)_p
\xrightarrow{\mathtt{call}(m,v)} ((m,E)::\eval,0,R,S,\pub',\abs,k)_o
%&\qquad\text{where $m \in \abs$ and (\textbf{PC})}
\\[1mm]
(\textsc{OQ})&
(\eval,l,R,S,\pub,\abs,k)_o
\xrightarrow{\mathtt{call}(m,v)} ((m,l+1)::\eval,mv,R,S,\pub,\abs',k)_p\text{ if }R(m)=\lambda x.M
  \\[1mm]
%&\qquad\text{where $m \in \pub$, $R(m)=\lambda x.M$ and (\textbf{OC})}
%
(\textsc{PA})&
((m,l)::\eval,v,R,S,\pub,\abs,k)_p
\xrightarrow{\mathtt{ret}(m,v)} (\eval,l,R,S,\pub',\abs,k)_o
  \\[1mm]
%&\qquad\text{where $m \in \pub$ and (\textbf{PC})}
(\textsc{OA})&
((m,E)::\eval,l,R,S,\pub,\abs,k)_o
\xrightarrow{\mathtt{ret}(m,v)} (\eval,E[v],R,\pub,\abs',k)_p
% &\qquad\text{where $m \in \abs$ and (\textbf{OC})}
    \\[-2mm] \text{}
    \end{tabular}\hrule\vspace{-3mm}
  \[
\!\!\!\!(\textbf{PC}):\; m\in\abs\land \pub' = \pub \cup (\meths(v) \cap dom(R)),\;
(\textbf{OC}):\; m\in\pub\land\abs' = \abs \cup (\meths(v) \setminus dom(R)).
\]\vspace{-5mm}
\caption{Trace semantics rules. Rules (PQ),\,(PA) assume the condition (\textbf{PC}), and similarly for (OQ),(OA) and (\textbf{OC}). $\meths(v)$\! contains all method names appearing in $v$. INT stands for \emph{internal} transition; PQ for \emph{$P$-question} (i.e.\ call);
PA for \emph{$P$-answer} (i.e.\ return). Similarly for OQ and OA.}\label{fig:games}\label{fig:traces}
\end{figure}

The trace semantics uses \emph{game configurations}, which are divided into \emph{$P$-configurations} and \emph{$O$-configurations} given respectively as:
\[
(\eval,M,R,S,\pub,\abs,k)_p \quad\text{ and }\quad (\eval,l,R,S,\pub,\abs,k)_o\,.
\]
In a $P$-configuration, a term $M$ is being evaluated\,--\,this is $P$'s role.
In an $O$-configuration, an external call has been made and the semantics waits for $O$ to either return that call, or reply itself with another call.
The components $M,R,S,\pub,\abs,k,l$ are as above, while $\eval$ is an \emph{evaluation stack}:
\[
\eval \ ::=\ \varepsilon \mid (m,E)::\eval \mid (m,l)::\eval
\]
which keeps track of the computations that are on hold due to external calls. The trace semantics is generated by the rules given in Figure~\ref{fig:games}.

The formulation follows closely the operational game semantics technique. For example, from a  $P$-configuration $(\eval,M,R,S,\pub,\abs,k)_p$, there are 3 options:
\begin{enumerate}
\item If $M$ can make an internal reduction, i.e.\ in the operational semantics in context $(R,S,k)$, then $(\eval,M,R,S,\pub,\abs,k)_p$ performs this reduction (via (INT)).
\item If $M$ is stuck at a method application for a method that is not in the repository $R$, then that method must be abstract (i.e.\ external) and needs to be called
  externally. This is achieved be issuing a call move and moving to an $O$-configuration (via (PQ)). The current evaluation context and the called method name are stored, in order to resume once the call is returned (via (OA)).
\item If $M$ is a value and the evaluation stack is non-empty, then $P$ has completed a method call that was issued by $O$ (via (OQ)) and can now return (via (PA)).
%If, on the other hand, the evaluation stack is empty, then no rule applies and the trace is ended.
\end{enumerate}
On the other hand, from an $O$-configuration $(\eval,l,R,S,\pub,\abs,k)_o$, there are 2 options: 
\begin{enumerate}
\item either return the last open method call (made by $P$) via (OA), or 
\item call one of the public methods (from $\pub$) using (OQ).
\end{enumerate}

The role of conditions (PC) and (OC) is to ensure that each player calls the methods owned by the other, or returns their own, and update the sets of public and abstract names according to the method names passed inside $v$. 

\begin{remark}
  The novelty of Figure~\ref{fig:games} with respect to previous work on trace semantics for open libraries (e.g.~\cite{DBLP:journals/corr/MurawskiT16a}) lies in the use of $l$ in order to bound the ability of $O$ to ask repeated questions for finite analysis. The way rules (OQ) and (PA) are designed is such that any sequence of consecutive  $O$-calls and $P$-returns has maximum length $2n$ if we bound $l$ to $n$ (i.e. $l \leq n$), as each such pair of moves increases $l$ by 1. On the other hand, each $P$-call supplies to $O$ a fresh counter ($l = 0$) to be used in contiguous (OQ)-(PA)'s. Thus, $l$ can be seen as keeping track of the \emph{insistence} of $O$ in calling.
\end{remark}

%\footnote{This reflects the fact that $O$ cannot distinguish between method names belonging to $P$ and, for example, $m$ and $\lambda x.mx$ are contextually equivalent terms. 

Finally, we can define the trace semantics of libraries. % and clients.
%We call a configuration \emph{final} if it is stuck: there is no next configuration from it.

\begin{definition}\label{def:tracesem}
  Let $L$ be a library. The semantics of $L$ is :%and $C$  respectively is:
  \begin{align*}
    \sem{L} &= \{ (\tau,\rho) \mid (L,\emptyset,\emptyset,\emptyset,\emptyset)\xrightarrow{bld}{\!\!}^*\ (\varepsilon,R,S,\pub,\abs)\land
    (\varepsilon,0,R,S,\pub,\abs,0)_o\xrightarrow{\tau}\rho\}
  \end{align*}
We say that $\sem{L}$ \emph{fails} if it contains some $(\tau,(\eval,E[\assert(0)],\cdots))$.
\end{definition}

\begin{example}\label{example:concrete}
Consider the DAO example as library $L_\mathtt{DAO}$ once again. Evaluating the game semantics we know the following sequence is in $\sem{L_\mathtt{DAO}}$. For economy, we hide $R,\pub,\abs$ and show only the top of the stack in the configurations. We also use $m(v)?$ and $m(v)!$ for calls and returns. We write $S_i$ for the store $[bal \mapsto i]$.
\begin{align*}
(\varepsilon,2&,S_{100},0)_o 
\xrightarrow{wdraw(42)?}((wdraw,1),wdraw(42),S_{100},0)_p\\
&\xrightarrow{}^*((wdraw,1),E[send(42)],S_{100},1)_p
\xrightarrow{send(42)?}((send,E),2,S_{100},1)_o\\
&\xrightarrow{wdraw(100)?}((wdraw,1),wdraw(100),S_{100},1)_p\\
&\xrightarrow{}^*((wdraw,1),E'[send(100)],S_{100},2)_p
\xrightarrow{send(100)?}((send,E),2,S_{100},2)_o\\
&\xrightarrow{send(())!}((wdraw,1),E'[()],S_{100},2)_p
\xrightarrow{}^*((wdraw,1),(),S_0,2)_p\\
&\xrightarrow{wdraw(())!}((send,E),1,S_0,2)_o
\xrightarrow{send(())!}((wdraw,1),E[()],S_0,1)_p\\
&\xrightarrow{}^*((wdraw,1),E[\assert(-42 \geq 0)],S_{-42},1)_p
\end{align*}
This transition sequence is an instance of the symbolic trace provided in the Introduction. Here, a call is made with parameter 42, and a reentrant call with 100, which leads to the assertion violation $\assert(-42 \geq 0)$. {Note that a bound of $k \leq 2$ is sufficient to find this assertion violation.}
\end{example}

We next establish two focal properties of the trace semantics: {bounding $k$ and $l$ ensures termination} (Theorem~\ref{termination}\shorten{}{, see Appendix~\ref{app:bounds}}), and that it is sound and complete with respect to library errors (Theorem~\ref{SandC}). %{More details are presented in Appendix~\ref{apx:SC}}.

%Compositionality in this library-client setting amounts to the fact that a linked library-client pair reaches a given final configuration iff the library and the client independently reach it with the same trace. In effect, compositionality establishes that the trace semantics is sound with respect to the operational semantics.

\begin{theorem}[Boundedness]\label{termination}
 {For any game configuration $\rho$, provided an upper bound $k_0$ and $l_0$ for call counters $k$ and $l$, the labelled transition system starting from $\rho$ is strongly normalising.}
\end{theorem}
\shorten{\begin{proof}
For any transition sequence $\rho=\rho_0 \to \dots \to \rho_i \to \dots$ and each $i>0$, we set the following two classes of configurations:
\[
(A) = \{\rho_i \mid \sizeof{\rho_{i}} < \sizeof{\rho_{i-1}}\}\qquad
(B) = \{\rho_i \mid \exists j < i-1.\ \sizeof{\rho_i} < \sizeof{\rho_j}\}
\]
where $\sizeof{\rho} = (k_0 - k,\sizeof{M},l_0 - l)$ is the \emph{size} of $\rho$, and $\sizeof{\rho} < \sizeof{\rho'}$ is defined by the lexicographic ordering of the triple $(k_0 - k,\sizeof{M},l_0 - l)$, with bounds $k_0$ and $l_0$ such that $k \leq k_0$ and $l \leq l_0$ for semantic transitions to be applicable. If not present in the configuration, we look at the evaluation stack $\eval$ to find the top-most missing component. In other words, opponent configurations will have size $(k_0 - k,\sizeof{E},l_0 - l)$ where $E$ is the top-most one in $\eval$, whereas proponent configurations will have size $(k_0 - k,\sizeof{M},l_0 - l)$ where $l$ is the top-most one in $\eval$.

We approach the proof two steps: (1)~we show that, for any transition sequence out of $\rho$, each reachable configuration belongs to (at least) one of the above classes; and (2)~prove that the classes form a terminating sequence.
For~(1), we do a case analysis on every transition in the game semantics and observe that the target configuration must belong to a class (see~\cite{FULL}). For~(2), let us assume there is an infinite sequence
\[\rho_0 \to \dots \to \rho_j \to \dots \to \rho_i \to \dots\]
Since all reachable configurations fall into either (A) or (B) class, we know that the sequence must comprise only (A) and (B) configurations. In this infinite sequence, we know that all sequences of (A) configurations are in descending size, so (A) sequences cannot be infinite. We also observe that (B) configurations are padded with (A) sequences. For instance, if $\rho_i$ is a (B) configuration, and $\rho_j$ is its matching configuration, there may exist     nested (B) configurations between $\rho_j$ and $\rho_i$, as well as (A) sequences padding these.

Additionally, these (B) configurations can only occur as a return to a call, so we know they only occur together with the introduction of evaluation boxes $\evalbox{\bullet}$. Since these brackets occur in pairs and are introduced in a nested fashion, we know $\eval$ can only contain evaluation contexts with well-bracketed evaluation boxes, meaning that there cannot be interleaved sequences of (B) configurations where their target configurations intersect. More specifically, the sequence 
\[\rho_0 \to \dots \to \rho_j \to \dots \to \rho_j' \to \dots \to \rho_i \to \dots \to \rho_i' \to \dots\]
where $\rho_i'$ matches $\rho_j'$ and $\rho_i$ matches $\rho_j$ is not possible.

Now, ignoring all (A) and nested (B) sequences, we are left with an infinite stream of top-level (B) sequences which are also in descending order. Since starting size is finite, we cannot have an infinite stream of (B) sequences. Thus, the assumption that the sequence is infinite does not hold, meaning our semantics is strongly normalising.
\end{proof}}{}

\begin{theorem}[S and C]\label{SandC}
{We call a client \emph{good} if it contains no assertions.
  For any library $L$, the following are equivalent:
  \begin{enumerate}
  \item $\sem{L}$ fails (reaches an assertion violation)
  \item there exists a good client $C$ such that $\sem{L \comp C}$ fails
  \end{enumerate}}
\end{theorem}
\begin{proof}
    1 to 2:~ Suppose now that $(\tau,\rho)\in\sem{L}$ for some trace $\tau$ and failed $\rho$. By Theorem~\ref{thm:definability}, we have that there is a good client $C$ realising the trace $\tau$. So then, by Lemma~\ref{simsim}, we have that $\sem{L\comp C}$ fails.
    
	2 to 1:~ Suppose $\sem{L \comp C}$ fails for some good client $C$. Then, by Lemma~\ref{simsim}, there are $\tau,\rho,\rho'$ such that $(\tau,\rho)\in\sem{L}$, $(\tau,\rho')\in\sem{C}$, and $\rho$ is failed (i.e.\ is of the shape $(\eval,E[\assert(0)],\cdots)$).
\end{proof}

The latter relies on an auxiliary lemma (well-composing of libraries and clients, \shorten{see~\cite{FULL}}{see Appendix~\ref{apx:SC}}), and a definability result akin to game semantics definability arguments (see Appendix~\shorten{\ref{app:def}}{\ref{app:def:full}}).

\begin{lemma}[L-C Compositionality]\label{simsim}
For any library $L$ and compatible good client $C$, $\sem{L \comp C}$ fails if and only if there exist $(\tau_1,\rho_1) \in \sem{L}$ and $(\tau_2,\rho_2) \in \sem{C}$ such that $\tau_1=\tau_2$ and $\rho_1=(\eval,E[\assert(0)],\cdots)$.
\end{lemma}

\begin{theorem}[Definability]\label{thm:definability}
	Let $L$ be a library and $(\tau,\rho)\in\sem{L}$. There is a good client $C$ compatible with $L$ such that $(\tau,\rho')\in\sem{C}$ for some $\rho'$.
\end{theorem}

%%% Local Variables:
%%% mode: latex
%%% TeX-master: "paper"
%%% End:

\section{Symbolic Semantics}
{Checking} libraries for errors using the semantics of the previous section is infeasible, even {when} the traces are bounded in length, as ground values are concretely represented.
In particular, integer values provided by $O$ as arguments to calls or return values range over all integers.
The typical way to mitigate this limitation is to execute the semantics symbolically, using symbolic variables for integers and path conditions to bind these variables to plausible values.
We use this technique to devise a symbolic version of the trace semantics, corresponding to a symbolic execution which will enable us in the next sections to introduce a practical method and implementation for {checking libraries for errors}. The symbolic semantics is fully formal, closely following the developments of the previous section, and allows us to prove a strong form of correspondence between concrete and symbolic semantics (a bisimulation).

Apart from integers, another class of concrete values provided by $O$ are method names. For them, the semantics we defined is symbolic by design: all method names played by $O$ are going to be fresh and therefore picking just one of those fresh choices is sufficient (formally speaking, the semantics lives in nominal sets~\cite{nom2}).
The reason why using fresh names for methods played by $O$ is sound is that the effect of $O$ calling a higher-order public method with an argument $m$ (where $m$ is another public method), and calling $\lambda x.mx$, is equivalent as far as reachability of an error is concerned. In the latter case, the client semantics would create a fresh name $m'$, bind it to $\lambda x.mx$, and pass $m'$ as an argument.
We therefore just focus on this latter case.

%From the operational semantics, we define rules for transition ($\symto$) for symbolic execution of a given higher-order open term configuration $(M,R,S,k)$ with free variables of integer type $\ints$, path condition $pc$, and symbolic environment $\sigma$. Since all free variables are of type $\ints$, it is the case that $M$ only has free variables of ground type. 

The symbolic semantics involves terms that may contain symbolic values for integers. We therefore extend the syntax for values and terms to include such values, and abuse notation by continuing to use  $M$ to range over them. 
We let $\sints$ be a set of symbolic integers ranged over by $\kappa$ and variants, and define:
\begin{align*}
  \mathit{Sym. Values}\quad \tilde{v}\ &::=\ m\mid i\mid ()\mid\kappa\mid \tilde v\oplus\tilde v\mid\pair{\tilde v,\tilde v}\\
  \mathit{Sym. Terms}\quad {M}\ &::=\ \cdots \mid\kappa
\end{align*}
where, in $\tilde v\oplus\tilde v$, not both $\tilde v$ can be integers.
We moreover use a symbolic environment to store symbolic values for references, but also to keep track of arithmetic performed with symbolic integers. More precisely, we let $\sigma$ be a finite partial map from the 
set $\mathtt{SInts} \cup \refs$ to symbolic values. Finally, we use $pc$ to range over program conditions, which will be quantifier-free first-order formulas
with variables taken from $\sints$, and with $\top,\bot$ denoting true and false respectively.

The semantics for closed symbolic terms involves configurations of the form $(M,R,\sigma,pc,k)$. Its rules include copies of those from Figure~\ref{fig:holib} (top) where the $pc$ and $\sigma$ are simply carried over. For example:
\[(E[\lambda x.M],R,\sigma,pc,k) \symto (E[m],R\uplus\{m\mapsto \lambda x.M\},\sigma,pc,k)\]
where $m$ is fresh.
On the other hand, the following rules directly involve symbolic reasoning:
\begin{flalign*}
%&(M,R,\sigma,pc,\nil) \symto (\nil,\sigma,pc)\\
&(E[\assert(\kappa)],R,\sigma,pc,k) \symto (E[\assert(0)],\sigma,pc \land (\kappa = 0),k)\\
&(E[\assert(\kappa)],R,\sigma,pc,k) \symto (E[()],R,\sigma,pc \land (\kappa \neq 0),k)\\
&(E[!r],R,\sigma,pc,k) \symto (E[\sigma(r)],R,\sigma,pc,k)\\
&(E[r:= \tilde v],R,\sigma,pc,k) \symto(E[()],R,\sigma[r\mapsto \tilde v],pc,k)\\
&(E[\tilde v_1 \oplus \tilde v_2],R,\sigma,pc,k) \symto (E[\kappa],R,\sigma\uplus\{\kappa \mapsto \tilde v_1 \oplus \tilde v_2\},pc,k)\quad\text{where $\kappa$ is fresh}\\
&(E[\ifm{\kappa}{M_1}{M_0}],R,\sigma,pc,k) \symto (E[M_0],R,\sigma, pc \land (\kappa = 0),k)\\
&(E[\ifm{\kappa}{M_1}{M_0}],R,\sigma,pc,k) \symto (E[M_1],R,\sigma, pc \land (\kappa \neq 0),k)
\end{flalign*}
and where $\tilde v_1\oplus\tilde v_2$ is a symbolic value (for $i_i\oplus i_2$ the rule from Figure~\ref{fig:holib} applies).

\begin{figure}[t]\small
\begin{tabular}{LL}
\begin{tabular}{C}(\widetilde{\textsc{INT}})\\[1mm]\text{} \end{tabular} &
\infer{(\eval,M,R,\pub,\abs,\sigma,pc,k)_p
                \symto (\eval,M',R',\pub,\abs,\sigma',pc',k')_p}
                {(M,R,\sigma,pc,k)\symto\ (M',R',\sigma,pc',k')}
\\[-1mm]
(\widetilde{\textsc{PQ}})&
(\eval,E[m\tilde v],R,\pub,\abs,\sigma,pc,k)_p
\xsymto{\mathtt{call}(m,\tilde v)} ((m,E)::\eval,0,R,\pub',\abs,\sigma,k)_o
%&\qquad\text{where $m \in \abs$ and (\textbf{PC})}
\\[1.5mm]
(\widetilde{\textsc{OQ}})&
(\eval,l,R,\pub,\abs,\sigma,pc,k)_o
\xsymto{\mathtt{call}(m,\tilde v)} ((m,l+1)::\eval,m\tilde v,R,\pub,\abs',\sigma,pc,k)_p %\text{ if }R(m)=\lambda x.M
  \\[1.5mm]
%&\qquad\text{where $m \in \pub$, $R(m)=\lambda x.M$ and (\textbf{OC})}
%
(\widetilde{\textsc{PA}})&
((m,l)::\eval,\tilde v,R,\pub,\abs,\sigma,pc,k)_p
\xsymto {\mathtt{ret}(m,\tilde v)} (\eval,l,R,\pub',\abs,\sigma,pc,k)_o
  \\[1.5mm]
%&\qquad\text{where $m \in \pub$ and (\textbf{PC})}
(\widetilde{\textsc{OA}})&
((m,E)::\eval,l,R,\pub,\abs,\sigma,pc,k)_o
\xsymto {\mathtt{ret}(m,\tilde v)} (\eval,E[\tilde v],R,\pub,\abs',\sigma,pc,k)_p
% &\qquad\text{where $m \in \abs$ and (\textbf{OC})}
    \\[-2mm] \text{}
    \end{tabular}\hrule
  \begin{tabular}{L@{\;\;\;}L}\\[-3mm]
(\widetilde{\textbf{PC}})& \text{$m\in\abs$ and $\pub' = \pub \cup (\meths(\tilde v) \cap dom(R))$.}
\\[1mm]
    (\widetilde{\textbf{OC}})& \text{$m\in\pub$ and $(\tilde v',\abs')\in\textsf{symval}(\theta,\abs)$ where $\theta$ is the expected type  of $\tilde v$. Moreover:}\\
&\mathtt{symval}(\theta,\abs) = 
\begin{cases} 
\{((),\abs)\} & \text{if $\theta = \mathtt{unit}$}\\
\{(\kappa,\abs\uplus\{\kappa\})\mid  \text{ $\kappa$ is fresh in $dom(\sigma) \uplus \abs$}\} & \text{if $\theta = \mathtt{int}$}\\
\{(m,\abs\uplus\{m\})\mid \text{$m$ is fresh in $dom(R) \uplus \abs$}  \} & \text{if $\theta = \theta_1 \to \theta_2$} \\
\{(\pair{\tilde v_1, \tilde v_2},\abs_2)\mid
(\tilde v_1,\abs_1)\in\mathtt{symval}(\theta_1,\abs)
& \text{if $\theta = \theta_1 \times \theta_2$}\\
\qquad\qquad\quad\qquad ( \tilde v_2,\abs_2)\in\mathtt{symval}(\theta_2,\abs_1)\}
\end{cases}\\[-3mm]
  \end{tabular}
\caption{Symbolic trace semantics rules. Rules (\textilde{PQ}),\,(\textilde{PA}) assume the condition (\textilde{\textbf{PC}}), and similarly for (\textilde{OQ}),(\textilde{OA}) and (\textilde{\textbf{OC}}).}\label{fig:symgames}
\end{figure}

We now extend the symbolic setting to the trace semantics.
We define symbolic configurations for $P$ and $O$ respectively as:
\[(\eval,M,R,\pub,\abs,\sigma,pc,k)_p\\
  \qquad\quad (\eval,l,R,\pub,\abs,\sigma,pc,k)_o
  \]
  with evaluation stack $\eval$, proponent term $M$, counters $k,l \in \mathbb{N}$, method repository $R$, public method name set $\pub$, $\sigma$ and $pc$ as previously.
The abstract name set $\abs$ is now a finite subset of $\meths\cup\sints$, as we also need to keep track of the symbolic integers introduced by $O$ (in order to be able to introduce fresh such names).
The rules for the symbolic trace semantics are given in Figure~\ref{fig:symgames}. {Note that $O$ always refreshes names it passes. This is a sound overapproximation of all names passed for the sake of analysis.}

Similarly to Definition~\ref{def:tracesem}, we can define the symbolic semantics of libraries.
\begin{definition}\label{def:bmc}
  Given library $L$, the symbolic semantics of $L$ is:
  \begin{align*}
    \sem{L}_s = \{ (\tau,\rho) \mid & (L,\emptyset,\emptyset,\emptyset,\emptyset)\xrightarrow{bld}{\!\!}^*\ (\varepsilon,R,S,\pub,\abs)\\
    &\land
    (\varepsilon,0,R,\pub,\abs,S,\top,0)_o\xrightarrow{\tau}_s\rho \ \land \
    \exists \mathcal{M}.\, \mathcal{M}\vDash \rho(\sigma)^\circ \land \rho(pc) \}
  \end{align*}
  where $\rho(\chi)$ is component $\chi$ in configuration $\rho$.
We say that $\sem{L}_s$ \emph{fails} if it contains some $(\tau,(\eval,E[\assert(0)],\cdots))$.
\end{definition}

The symbolic rules follow those of the concrete semantics, the biggest change being the treatment of symbolic values played by $O$. Condition ($\widetilde{\textbf{OC}}$) stipulates that $O$ plays distinct fresh symbolic integers as well as fresh method names, in each appropriate position in $\tilde v$, and all these names are included in the set $\abs$.

\begin{example}
As with Example~\ref{example:concrete}, we consider the DAO attack. Running the symbolic semantics, we find the following minimal class of errors. We write $\sigma_{\tilde v}$ for a symbolic environment $[bal \mapsto \tilde v]$.
\begin{align*}
(\varepsilon,2,\sigma&_{100},k_0)_o \xrightarrow{wdraw(\kappa_1)?}((wdraw,1),wdraw(\kappa_1),\sigma_{100},2)_p\\
&\xrightarrow{}^*((wdraw,1),E[send(\kappa_1)],\sigma_{100},1)_p
\xrightarrow{send(\kappa_1)?}((send,E),2,\sigma_{100},1)_o\\
&\xrightarrow{wdraw(\kappa_2)?}((wdraw,1),wdraw(\kappa_2),\sigma_{100},1)_p\\
&\xrightarrow{}^*((wdraw,1),E'[send(\kappa_2)],\sigma_{100},0)_p
\xrightarrow{send(\kappa_2)?}((send,E),2,\sigma_{100},0)_o\\
&\xrightarrow{send(())!}((wdraw,1),E'[()],\sigma_{100},0)_p\\
&\xrightarrow{}^*((wdraw,1),(),\sigma_{100-\kappa_2},0)_p\xrightarrow{wdraw(())!}((send,E),1,\sigma_{100-\kappa_2},0)_o\\
&\xrightarrow{send(())!}((wdraw,1),E[()],\sigma_{100-\kappa_2},1)_p\\
&\xrightarrow{}^*((wdraw,1),E[\assert(!bal \geq 0)],\sigma_{100-\kappa_2-\kappa_1},1)_p
\end{align*}
For this to be a valid error, we require $(\kappa_1,\kappa_2 \leq 100) \land (100 - \kappa_2 - \kappa_1 < 0)$ to be satisfiable. Taking assignment $\{\kappa_1 \mapsto 100, \kappa_2 \mapsto 1\}$, we show the path is valid.
\end{example}

\subsection{Soundness}
{
The main result of this section is establishing the soundness of the symbolic semantics: a trace and a specific configuration can be achieved symbolically iff they can be achieved concretely as well. In fact, we will need to quantify this statement as, by construction, the symbolic semantics requires $O$ to always place fresh method names, whereas in the concrete semantics $O$ is given the freedom to play old names as well.
What we show is that the symbolic semantics corresponds (via \emph{bisimilarity}) to a restriction of the concrete semantics where $O$ plays fresh names only. This restriction is sound, in the sense that it is sufficient for identifying when a configuration can fail.
We make this precise below.}

A \boldemph{model} $\cal M$ is a finite partial map from symbolic integers to concrete integers. Given such an $\cal M$ and a formula $\phi$, we define $\mathcal{M}\models\phi$ using a standard first-order logic interpretation with integers and arithmetic operators (in particular, we require that all symbolic integers in $\phi$ are in the domain of $\cal M$).
Moreover, for any symbolic term $M$ (or trace, move, etc.), we denote by $M\{\mathcal{M}\}$ the concrete term we obtain by substituting any symbolic integer $\kappa$ of $M$ with its corresponding concrete integer $\mathcal{M}(\kappa)$. 
Finally, given a symbolic environment $\sigma$, we define its formula representation $\sigma^\circ$ recursively by:
\[
  \emptyset^\circ=\top,\quad
  (\sigma\uplus\{r\mapsto v\})^\circ = \sigma^\circ,\quad
(\sigma\uplus\{\kappa\mapsto v\})^\circ = \sigma^\circ\land (\kappa = v).
\]

We now define notions for equivalence between symbolic and concrete configurations. Let $\cal M$ be a model. For any concrete configuration $\rho=(\eval,\chi,R,S,\pub,\abs,k)$ and symbolic configuration $\rho_s=(\eval',\chi',R',\pub',\abs',\sigma,pc,k')$, we say they are \emph{equivalent in $\cal M$}, written $\rho =_{\mathcal{M}} \rho_s$, if:
\begin{itemize}
\item $(\eval,\chi,R)=(\eval',\chi',R')\{\mathcal{M}\}, \pub=\pub',\abs=\abs'\cap\meths$ and  $S=(\sigma\upharpoonright\refs)\{\mathcal{M}\}$;
\item $\dom(\mathcal{M})=(\abs'\cup\dom(\sigma))\cap\sints$ and $\mathcal{M}\vDash pc \land \sigma^\circ$.
\end{itemize}
The notion of equivalence we require between concrete configurations and their symbolic counterparts is behavioural equivalence, {modulo $O$ playing fresh names}.

More precisely,
{a transition $\rho\xrightarrow{\chi}\rho'$ is called \emph{O-refreshing} if, when $\rho$ is an $O$-configuration and $\chi=\mathtt{call}/\mathtt{ret}(m,v)$ then all names in $v$ are fresh and distinct.}
A finite set 
 $\mathcal{R}$ with elements of the form $(\rho,\mathcal{M},\rho_s)$ is a \boldemph{bisimulation} if, whenever $(\rho,\mathcal{M},\rho_s)\in\cal R$, written $\rho\,\mathcal{R_M}\,\rho_s$ then
 $\rho =_\mathcal{M} \rho_s$ and, using $\chi$ to range over moves and $\varepsilon$ (i.e.\ no move):
\begin{itemize}
\item if $\rho \xrightarrow{\chi} \rho'$
{is $O$-refreshing}
  then there exists $\mathcal{M}' \supseteq \mathcal{M}$ such that $\rho_s\xsymto{\chi_s}\rho_s'$, with $\chi=\chi_s\{\mathcal{M}'\}$, and $\rho'\mathcal{R_{M'}}\rho_s'$;
\item if $\rho_s \xsymto{\chi} \rho_s'$ then there exists $\mathcal{M}' \supseteq \mathcal{M}$ such that $\rho\xgamto{\chi\{\mathcal{M}'\}}\rho'$ and $\rho'\mathcal{R_{M'}}\rho_s'$.
\end{itemize}
We let $\sim$ be the largest bisimulation relation: $\rho\sim_{\cal M}\rho_s$ iff there is bisimulation $\cal R$ such that $\rho\mathcal{R_M}\rho_s$. 

We can show that concrete and symbolic configurations are bisimilar.
\begin{lemma}\label{sound_bisim}
Given $\rho,\rho_s$ a concrete and symbolic configuration respectively, and $\cal M$  a model such that $\rho=_{\cal M}(\rho')$, we have $\rho\sim_{\cal M}\rho_s$.
\end{lemma}
\begin{proof}[Proof (sketch)]
We show that  $\{(\rho,\mathcal{M},\rho')\mid \rho=_{\cal M}\rho'\}$ is a bisimulation.
\end{proof}

{Next, we argue that $O$-refreshing transitions suffice for examining failure of concrete configurations. Indeed, suppose $\tau$ is a trace leading to fail, and where $O$ plays an old name $m$ in argument position in a given move. Then, $\tau$ can be simulated by a trace $\tau'$ that uses a fresh $m'$ in place of $m$.
If $m$ is an $O$-name, we obtain $\tau'$ from $\tau$ by following exactly the same transitions, only that some $P$-calls to $m$ are replaced by calls to $m'$ (and accordingly for returns). If, on the other hand, $m$ is a $P$-name, then the simulation performed by $\tau'$ is somewhat more elaborate: some internal calls to $m$ will be replaced by $P$-calls to $m'$, immediately followed by the required calls to $m$ (and dually for returns).}

\begin{lemma}[O-Refreshing]\label{lem:Orefresh}
Let $\rho$ be a concrete configuration. Then, $\rho$ fails iff it fails using only $O$-refreshing transitions.
\end{lemma}

With the above, we can prove soundness.

\begin{theorem}[Soundness]\label{sound_symb}
For any $L$, $\sem{L}$ fails iff $\sem{L}_s$ fails.
\end{theorem}
\begin{proof}
Lemma~\ref{sound_bisim} implies that $\sem{L}_s$ fails iff $\sem{L}$ fails with $O$-refreshing transitions, which in turns occurs iff $\sem{L}$ fails, by Lemma~\ref{lem:Orefresh}.
\end{proof}

%\nttodo{up to here}

%%% Local Variables:
%%% mode: latex
%%% TeX-master: "paper"
%%% End:

%\section{Bounded Model Checking Libraries} 

%Recall Definition~\ref{def:tracesem}, which provides the trace semantics for libraries and clients. In a similar way, we now define their symbolic semantics, which consists of the set of \emph{valid} symbolic configurations with their corresponding trace.

\subsection{Bounded Analysis for Libraries}

Definition~\ref{def:bmc} states how the symbolic trace semantics can be used to independently {check libraries for errors}. As with the trace semantics in Definition~\ref{def:tracesem}, this is strongly normalising when given an upper limit to the call counters. {As such, $\sem{L}_s$ with counter bounds $k_0,l_0 \in \mathbb{N}$, for $k,l$ respectively}, defines a finite set (modulo selecting of fresh names) of reachable valid configurations within $k \leq k_0, l \leq l_0$, where validity is defined by the satisfiability of the symbolic environment $\sigma$ and the path condition $pc$ of the configuration reached.
By virtue of Theorems~\ref{sound_bisim} and~\ref{SandC},
every valid reachable configuration that is failed (evaluates an invalid assertion) is realisable by some client. And viceversa.

Given a library $L$, taking $\mathcal{F}\sem{L}_s$ to be all reachable final configurations, we have the exhaustive set of paths $L$ can reach. In $\mathcal{F}\sem{L}_s$, every failed configuration $(\tau,\rho)$, i.e.\ such that $\rho$ holds a term $E[\assert(0)]$, defines a reachable assertion violation, where $\tau$ is a true counterexample. Hence, to check $L$ for assertion violations it suffices to produce a finite representation of the set $\mathcal{F}\sem{L}_s$. One approach is to bound the depth of analysis by setting an upper bound to the call counters, using a name generator to make deterministic the creation of fresh names, and then exhaustively search all final configurations for failed elements. In the following section we implement this routine and test it.

%%% Local Variables:
%%% mode: latex
%%% TeX-master: "paper"
%%% End:

\section{Implementation and Experiments}\label{sec:experiments}
We implemented the syntax and symbolic trace semantics (symbolic games) for {\holib} in the $\kfw$ semantic framework~\cite{k:framework} as a proof of concept, and tested it on 70 sample libraries.\footnote{The tool and its benchmarks can be found at: 
%\begin{center}
\url{https://github.com/LaifsV1/HOLiK}.}
%\end{center}
%
Using $\kfw$'s option to exhaustively expand all transitions, $\kfw$ is able to build a closure of all applicable rules. By providing a bound on the call counters, we produce a finite set of all reachable valid symbolic configurations up to the given depth (equivalent to finding every valid $\rho \in \mathcal{F}\sem{L}_s$) which thus implements our bounded symbolic execution.

%\subsection{Experiments}
We wrote and adapted examples of coding errors into a set of 70 sample libraries written in {\holib}, totalling 6,510 lines of code (LoC). Examples adapted from literature include: reentrancy bugs from smart contracts~\cite{smart:surveyAttacks,smart:oyente}; 
%the programs we saw in Section~2;
variations of the ``awkward example''~\cite{akward}; 
{various programs from the {\mochi} benchmark~\cite{mochi}}; 
and simple implementations related to concurrent programming (e.g. flat combining and race conditions) where errors may occur in a single thread due to higher-order behaviour.
We also combined several libraries, by concatenating refactored method and reference definitions, to generate larger libraries that are harder to solve. 
%In our benchmark, a filename $\texttt{X-mo}$ describes a program $\texttt{X}$ that has been extended, via simple concatenation, with all the method definitions inside our selection of $\mochi$ programs. These files are approximately 150 LoC each. The benchmark also contains filenames that include the term ``$\texttt{various}$''. 
%, with the largest being a combination of all programs. 
Combined files range from 150 to 520 LoC.

We ran {\holik} on all sample libraries, lexicographically increasing the bounds from $k\leq 2,l\leq 1$ to $k\leq 5,l\leq 3$ (totalling 78,120 LoC checked), with a timeout set to five minutes per library. We start from $k\leq 2$ because it provides the minimum nesting needed to observe higher-order semantics. All experiments ran on an Ubuntu 19.04 machine with 16GB RAM, Intel Core i7 3.40GHz CPU, with intermediate calls to \textsc{Z3} to prune invalid configurations. Per bound, the number of counterexamples found, the time taken in seconds, and the execution status, i.e. whether it terminated or not, are recorded in Table~\ref{tab:holik}.

\begin{table}
	\begin{center}
		\begin{tabular}{c | c | c | c }
& $l \leq 1$ & $l \leq 2$ & $l \leq 3$ \\
\hline
$k \leq 2$~ & 226/70/45 (555s) & 5708/60/44 (4710s) & 9656/3/23 (12471s) \rule{0pt}{4.5mm}\\[2mm]
$k \leq 3$~ & 1254/67/51 (1475s) & 4092/27/18 (13482s) & 4187/17/12 (16649s) \\[2mm]
$k \leq 4$~ & 3392/63/48 (3180s) & 3069/19/14 (15903s) & 1335/12/10 (17765s) \\[2mm]
$k \leq 5$~ & 3659/57/45 (4787s) & 895/15/10 (16757s) & 215/11/9 (17796s)\\
%$k \leq 2$ & 226/70 (9m 15s) & 5708/60 (1h 18m 30s) & 9656/3 (3h 27m 51s) \\
%$k \leq 3$ & 1254/67 (24m 35s) & 4092/27 (3h 44m 42s) & 4187/17 (4h 37m 29s) \\
%$k \leq 4$ & 3392/63 (53m) & 3069/19 (4h 25m 3s) & 1335/12 (4h 56m 5s) \\
%$k \leq 5$ & 3659/57 (1h 19m 47s) & 895/15 (4h 39m 17s) & 215/11 (4h 56m 36s)
		\end{tabular}\\
	\text{}\\
	$a$/$b$/$c$ ($d$) for $a$ traces found in $b$ successful runs taking $d$ seconds in total\\
	where $c$ out of 59 unsafe files were found to have bugs, per bound.\\
	59 of 59 unsafe files found to have bugs over the various bounds checked
	\end{center}
	\caption{Table recording performance of \holik{} on our benchmarks}
	\label{tab:holik}
\end{table}

%\subsection{Observations and Results}
We can observe that independently increasing the bounds for $k$ and $l$ causes exponential growth in the total time taken, which is expected from symbolic execution. Note that the time tends towards 21000 seconds because of the timeout set to 5 minutes for 70 programs. The number of errors found also grows exponentially with respect to the increase in bounds, which can be explained by the exponential growth in paths. With bounds $k \leq 2$ and $l \leq 1$, all 70 programs in our benchmark were successfully analysed, though not all minimal errors were found until the bounds were increased further. Cumulatively, all unsafe programs in our benchmark were correctly identified.

While the table may suggest that increasing bound for $l$ is more beneficial than that for $k$, the number of errors reported does not imply every trace is useful. For instance, increasing the bound for $l$ can lead to errors re-merging in a higher-order version, which suggests potential gain from a partial order reduction. Overall, the $k$ and $l$ counters are incomparable as they keep track of different behaviours. Finally, since $\holik$ was able to handle every file and correctly identified all unsafe files in the benchmark, we conclude that $\holik$, as a proof of concept, captures the full range of behaviours in higher-order libraries. Results suggest that the tool scales up to at least medium-sized programs (<1000 LoC), which is promising because real-world medium-size higher-order programs have been proven infeasible to check with standard techniques (e.g. the DAO withdraw contract was approximately 100 LoC).

%%% Local Variables:
%%% mode: latex
%%% TeX-master: "paper"
%%% End:

\section{Related Work}

Game semantics techniques have been applied to program equivalence verification by reducing program equivalence to language equivalence in a decidable automata class~\cite{AGS1,AGS2}. 
Equivalence tools can be used for reachability but, as they perform full verification, they can only cover lower-order recursion-free language fragments. 
For example, the {\coneqct}~\cite{coneqct} tool can verify the simplified DAO attack, but cannot check higher-order or recursive functions (e.g.\ the ``file lock" and ``flat combiner" examples), and {operates on} integers concretely.
%It also uses concrete integers for which one must define a bound, while we have unbounded and symbolic integers.
Close to our approach is also Symbolic GameChecker~\cite{Dimovski}, which
performs symbolic model checking by using a representation of games based on symbolic finite-state automata. The tool works on recursion-free Idealized Algol with first-order functions, which supports only integer references. On the other hand, it is complete (not bounded)  on the fragment that it covers.

Besides games techniques, a recent line of work on verification of contracts in \racket{} \cite{scv:relatively:complete, scv:total:verification} is the work closest to ours. \racket{} contracts exist in a higher-order setting similar to ours, and generalise higher-order pre and post conditions, and thus specify safety. To verify these, \cite{scv:relatively:complete} defines a symbolic execution based on what they call ``{demonic context}'' in prior work~\cite{scv:demonic:context}. This either returns a symbolic value to a call, or performs a call to a known method within some unknown context, thus approximating all the possible higher-order behaviours, and is equivalent to the role the opponent plays in our games. In \cite{scv:total:verification}, the technique is extended to handle state, and finitised for total verification. The approaches are notionally similar to ours, since both amount to Symbolic Execution for an unknown environment. In substance, the techniques are very different and in particular ours is based on a semantics theory which allows us to obtain compositionality and definability results. On the other hand, \racket{} contracts can be used for richer verification questions than assertion violations. In terms of tool performance, we provide a comparison of the techniques in Appendix~\ref{app:racket}.

Another relevant line of work is that of verifying programs in the Ethereum Platform. Smart contracts call for techniques that handle the environment, with a focus on reentrancy. 
Tools like {\oyente}~\cite{smart:oyente} and {\majan}~\cite{smart:majan} use pre-defined patterns to find bugs in the transaction order, but are not sound or complete. 
{\reguard}~\cite{smart:reguard} finds sound reentrancy bugs using a fuzzing engine to generate random transactions to check with a reentrancy automaton. In principle, it may detect reentrancy faster than symbolic execution (native execution is faster \cite{fuzzers:vs:symbolic}), but, is incomplete even in a bounded setting. More closely related to our approach, \cite{smart:foundationsTools} considers the possibility of an \emph{unknown} contract $c?$ calling a \emph{known} contract $c*$ at each higher call level. This can be generalised in our game semantics as \emph{abstract} and \emph{public} names calling each other, but their focus is on modelling reentrancy, while we handle the full range of higher-order behaviours.

Like {\klee}~\cite{klee} and {\jcute}~\cite{jcute}, our implementation is a symbolic execution tool. These are generally able to find first-order counterexamples, but are unable to produce higher-order traces involving unknown code. Particularly, {\klee} and {\jcute} only handle symbolic calls provided these can be concretised. This partially models the environment, but calls are often impossible to concretise with libraries.
The {\cbmc}~\cite{DBLP:conf/tacas/ClarkeKL04,cprover:cbmc} bounded model checking approach, which also bounds function application to a fixed depth, partially handle calls to unknown code by returning a non-deterministic value to such calls. This is equivalent to a game where only move available to the opponent is to answer questions. This restriction allows {\cbmc} to find some bugs caused by interaction with the environment, but misses errors that arise from transferring flow of control (e.g. reentrancy). The typical BMC approach also misses bugs involving disclosure of names.

Higher-order model checking tools like $\mochi$~\cite{mochi} are also related. $\mochi$ model checks a pure subset of OCaml and is based on predicate abstraction and CEGAR and higher-order recursion scheme model checkers. The modular approach~\cite{mochi:modular} further extends this idea with modular analysis that guesses refinement intersection types for each top-level function. Although generally incomparable, $\holik$ covers program features that $\mochi$ does not: $\mochi$ does not handle references and support for open code is limited (from experiments, and private communication with the authors).

\section{Future Directions} Observing errors resurface deeper in the trace suggests the possibility of defining a partial order for our semantics to obtain equivalence classes for configurations and thus eliminate paths that involve known errors~\cite{partial-order:ample,partial-order:stubborn}. Additionally, while $k$ and $l$ successfully bound infinite behaviour, a notion of bounding can be arbitrarily chosen. In fact, while we chose to directly bound the sources of infinite behaviour in method calls for simplicity of proofs and implementation, the theory does not prevent the generalisation of $k$ and $l$ as a monotonic cost function that bounds the semantics. It may also be worth considering the elimination of bounds entirely for the sake of unbounded verification. For this, one direction is abstract interpretation~\cite{abs:int:cousot,abs:int:wide-narrow}, which amounts to defining overapproximations for values in our language to then attempt to compute a fixpoint for the range of values that assertions may take. However, defining and using abstract domains that maintain enough precision to check higher-order behaviours, such as reentrancy, is not a simple extension of the theory. Another direction, similar to $\coneqct$ \cite{coneqct}, is to define a push-down system for our semantics. Particularly, the approach in \cite{coneqct} is based on the decidability of reachability in fresh-register pushdown automata, and would require overapproximations for methods and integers. As with abstract interpretation, this would require defining abstract domains for methods and integers. While methods could be approximated using a finite set of names, as with $k$-CFA~\cite{kcfa}, an extension using integer abstract domains would need refinement to tackle reentrancy attacks. Finally, $\mochi$~\cite{mochi} shows that it is possible to use CEGAR and higher-order recursion schemes for unbounded verification of higher-order programs. However, an extension of the $\mochi$ approach to include references and open code is not obvious.

% \section{Conclusion}
% We have presented an approach based on game semantics for symbolic execution of higher-order open libraries, and implemented a prototype symbolic execution tool for higher-order libraries as a proof of concept. We tested our tool on a custom benchmark of medium-sized programs (<1000 LoC) that exceeds the state-of-the-art in terms of higher-order behaviours involving open code. Our experiments show that our prototype can handle the full range of open-code behaviours observed in the benchmark, and thus exceeds the state-of-the-art. While scalability has only been tested up to medium-sized programs, we believe there is a practical motivation for the approach as the real-world counterparts to the benchmark programs are of similar size and have proven difficult if not infeasible to verify using standard techniques. Additionally, even as a proof of concept, we believe the approach is promising as it handles a range of behaviours that existing tools do not cover.

% We also supported our approach with theoretical results that prove soundness of our analysis, as well as the soundness and completeness of our game semantics. For this, we have provided intermediate results for definability and library-client compositionality of the semantics, as well as a proof of boundedness of our analysis. Considering future directions, we believe that the game semantics approach is theoretically promising as a framework for analysis of (higher-order) open code.

%%% Local Variables:
%%% mode: latex
%%% TeX-master: "paper"
%%% End:

\bibliography{symbolic}

\newpage

\appendix 
\section{Motivating examples}\label{app:examples}

Our {file lock} example provides a scenario where the library makes it possible for the client to update a file without first reacquiring the lock for it. The library contains an empty private method \texttt{updateFile} that simulates file access. The library also provides a public method \texttt{openFile}, which locks the file, allows the user to update the file indirectly, and then releases the lock.
\begin{lstlisting}[label=filelock]
  import userExec :((unit -> unit) -> unit)
  int lock := 0;
  private updateFile(x:unit) :(unit) = { () };
  public openFile (u:unit) :(unit) = {
    if (!lock) then ()
    else (lock := 1;
         let write = fun(x:unit):(unit) -> (assert(!lock);updateFile()) 
         in userExec(write); lock := 0) };
\end{lstlisting}
The bug here is that \texttt{openFile} creates a \texttt{write} method, which it then passes to the client, via \texttt{userExec(write)}, to use whenever they want. This provides the client indirect access to the private method \texttt{updateFile}, which it can call without first acquiring the lock. Running this example in {\holik} we obtain the following minimal trace:
\begin{align*}
&call\pair{openFile,()} \cdot call\pair{userExec,m_2} \cdot ret\pair{userExec,()}\\
&\qquad\cdot ret\pair{openFile,()} \cdot call\pair{m_2,()}
\end{align*}
where $m_2$ is the method \emph{name} generated by the library and bound to the variable \texttt{write}. This example serves as a representative of a class of bugs caused by revealing methods to the environment, a higher-order problem, in this case involving the second-order method \texttt{userExec} revealing $m_2$.

Next, we simulate double deallocation using a global reference \texttt{addr} as the memory address. The library defines private methods \texttt{alloc} and \texttt{free} to simulate allocation and freeing. The empty private method \texttt{doSthing} serves as a placeholder for internal computation that does not free memory.
\begin{lstlisting}[label=dfree]
  import getInput :(unit -> int)
  int addr := 0; // 0 means address is free
  private alloc (u:unit) :(unit) = { 
    if not(!addr) then addr := 1 else () };
  private free (u:unit) :(unit) = { 
    assert(!addr); addr := 0 };
  private doSthing (i:int) :(unit) = { () };
  public run (u:unit) :(unit) = {
    alloc(); doSthing(getInput ()); free() };
\end{lstlisting}
The error occurs in line 9, which calls the client method \texttt{getInput}. This passes control to the client, who can now call \texttt{run} again, thus causing \texttt{free} to be called twice. Executing the example on {\holik}, we obtain the following trace:
\begin{align*}
&call\pair{run,()} \cdot call\pair{getInput,()} \cdot call\pair{run,()}\cdot call\pair{getInput,()} \\
&\qquad\cdot ret\pair{getInput,x_1} \cdot ret\pair{run,()} \cdot ret\pair{getInput,x_2}
\end{align*}
As with the DAO attack, this is a reentrancy bug.

Finally, we have an unsafe implementation of a flat combiner. The library defines two public methods: \texttt{enlist}, which allows the client to add procedures to be executed by the library, and \texttt{run}, which lets the client run all procedures added so far. The higher-order global reference \texttt{list} implements a list of methods.
\begin{lstlisting}[label=dfree]
  private empty(x:int) : (unit) = { () };
  fun list := empty;
  int cnt := 0; int running := 0;
  public enlist(f:(unit -> unit)) :(unit) = {
    if (!running) then ()
    else 
      cnt := !cnt + 1;
      (let c = !cnt in let l = !list in 
       list := (fun(z:int):(unit) -> if (z == c) then f() else l(z)))};
  public run(x:unit) :(unit) = {
    running := 1;
    if (0 < !cnt) then 
      (!list)(!cnt);
      cnt := !cnt - 1; assert(not (!cnt < 0)); run()
    else (list := empty; running := 0) };
\end{lstlisting}
The bug here is also due to a reentrant call in line~13. However, this is a much tougher example as it involves a higher-order reference \texttt{list}, a recursive method \texttt{run}, and a second-order method \texttt{enlist} that reveals client names to the library. With {\holik}, we obtain the following minimal counterexample:
\begin{align*}
&call\pair{enlist,m_1} \cdot ret\pair{enlist,()} \cdot call\pair{run,()}\cdot call\pair{m_1,()} \\
&\qquad\cdot call\pair{run,()} \cdot call\pair{m_1,()} \cdot ret\pair{m_1,()} \cdot ret\pair{run,()} \cdot ret\pair{m_1,()}
\end{align*}
where $m_1$ is a client name revealed to the library. In the trace above, \texttt{enlist} reveals the method $m_1$ to the library. This name is then added to the list of procedures to execute. In \texttt{run}, the library passes control to the client by calling $m_1$. At this point, the client is allowed to call \texttt{run} again before the list is updated.

%%% Local Variables:
%%% mode: latex
%%% TeX-master: "paper"
%%% End:

\section{Comparison with \racket{} Contract Verification}\label{app:racket}

We shall consider the latest version of the tool \cite{scv:total:verification} since it handles state, which we refer to as \scv{} (Software Contract Verifier). A small benchmark (19 programs) based on \holik{} and \scv{} benchmarks was used for testing. Programs were manually translated between \holib{} and \racket{}. Care was taken to translate programs whilst maintaining their semantics: contracts enforcing an input-output relation were translated into \holib{} using wrapper functions that define the relation through an if statement.
% not sure to include or not...
In the other direction, since contracts do not directly access references inside a term, stateful functions were translated from HOLi to return any references we wish to reason about.

Table~\ref{tab:holik:scv} records the comparison. On one hand, \holik{} only found real errors, whereas \scv{} reported several spurious errors--a third of all errors were spurious. On the other hand, \scv{} was able to prove total correctness of 3 of the 7 safe files present. \scv{} also scales much better than \holik{} with respect to program size, which is in exchange of precision. The difference in time for small programs is mainly due to initialisation time. 
Subtle differences in the nature of each tool can also be observed. e.g., \holik{} reports 1 real error for \texttt{ack-simple-e}, whereas \scv{} reports 2 errors. The difference is because \scv{} takes into account constraints for integers (e.g. $>0$ and $=0$). More interestingly, for \texttt{various}, \holik{} reports 19 ways to reach assertion violations, whereas \scv{} reports only 6 real ways to violate contracts. The difference is because \holik{} reports paths through the execution tree that reach errors, whereas \scv{} reports a set of terms that may violate the contracts. For instance, independently safe methods $A$ and $B$ that may call an unsafe method $C$ would be, from testing, reported as three valid traces ($call\pair{A} \cdot call\pair{C}$, $call\pair{B} \cdot call\pair{C}$ and $call\pair{C}$) by \holik{}. In contrast, \scv{} reports a single contract violation blaming $C$. Finally, \texttt{ack} failed to run on \scv{} due to unknown errors; \racket{} reported an error internal to the tool. Further testing proved the file is a valid \racket{} program that can be executed manually.

\begin{table}[t]\small
	\begin{center}
		\begin{tabular}{c | c c c | c c c c}
			Program & LoC & Traces & Time (s) & LoC & Errors & Time (s) & False Errors\\
			\hline
			ack  & 17 & 0 & 6.0 & 9 & N/A & 2.4 & N/A\\
			ack-simple  & 13 & 0 & 6.5 & 9 & 0 & 2.4 & 0\\
			ack-simple-e  & 13 & 1 & 6.5 & 9 & 2 & 2.5 & 0\\
			dao  & 10 & 0 & 5.0 & 15 & 1 & 2.6 & 1\\
			dao-e  & 16 & 1 & 5.5 & 15 & 1 & 2.7 & 0\\
			dao-various  & 85 & 5 & 22.5 & 122 & 10 & 3.0 & 5\\
			dao2-e  & 85 & 10 & 23.5 & 122 & 10 & 2.9 & 0\\
			escape  & 9 & 0 & 5.0 & 9 & 0 & 2.6 & 0\\
			escape-e  & 9 & 2 & 5.0 & 10 & 1 & 2.7 & 0\\
			escape2-e  & 10 & 14 & 6.0 & 10 & 1 & 2.7 & 0\\
			factorial  & 10 & 0 & 5.0 & 9 & 0 & 2.2 & 0\\
			mc91  & 12 & 0 & 5.0 & 9 & 1 & 2.2 & 1\\
			mc91-e  & 12 & 1 & 5.0 & 8 & 1 & 2.4 & 0\\
			mult  & 14 & 0 & 5.0 & 11 & 2 & 2.7 & 2\\
			mult-e  & 14 & 1 & 5.0 & 11 & 2 & 2.4 & 0\\
			succ  & 7 & 0 & 5.0 & 7 & 1 & 2.5 & 1\\
			succ-e  & 7 & 1 & 5.0 & 7 & 1 & 2.8 & 0\\
			various & 116 & 19 & 14.0 & 108 & 11 & 6.2 & 5\\
			total & 459 & 55 & 140.5 & 500 & 45 & 49.8 & 15
		\end{tabular}
	\end{center}
	\caption{Comparison of \holik{} (left) and \scv{} (right). N/A is recorded for \texttt{ack} as in our attempts \scv{} crashed due to unknown reasons.}
	\label{tab:holik:scv}
\end{table}

%%% Local Variables:
%%% mode: latex
%%% TeX-master: "paper"
%%% End:

\section{ML-like References}\label{general:refs}

\holib\ has global higher-order references. These are enough for coding all of our examples and, moreover, allow us to prove {completeness} (every error has a realising client).
We here
present a sketch of how games can be extended with (locally created, scope extruding) ML-like references, following e.g.~\cite{OGS2,OGS3}.
First, the following extension to types and terms are required.
\begin{align*}
\theta ::= \cdots \mid \mathtt{ref}~\theta &&
  M ::= \dots \mid {!} M \mid \mathtt{ref}~M\mid M=M && v ::= \dots \mid r
\end{align*}
The term ${!} M$ allows dereferencing terms $M$ which evaluate to references, while $\mathtt{ref}~v$ creates dynamically a fresh name $r \in \refs_{\theta}$ (if $v:\theta$), and the semantic purpose is to update the store $S\uplus\{r \mapsto v\}$ when evaluating $\mathtt{ref}~v$.
Note that this allows us to store references to references, etc.
Finally, the construct $M=M$ is for comparing references for name equality.
%In a symbolic setting, $\mathtt{ref}~\kappa$ updates $\sigma$ instead.

With terms handling general references concretely and symbolically, we extend game configurations with sets $\mathcal{L}_p,\mathcal{L}_o \subseteq \refs$ that keep track of reference names disclosed by the proponent and opponent respectively.
References being passed as values means that the client can update the references belonging to the client, and viceversa.
When making a move, for each reference $r$ they own that is passed, the proponent adds $r$ to $\mathcal{L}_p$.
Passing of names in a move can be done either by method argument and return value, but also via the common part of the store (i.e.\ via the references known to both players). Similarly, opponent passes names in their moves, which are added to $\mathcal{L}_o$. Concretely, when the opponent passes control, all references in $\mathcal{L}_p$ are updated with opponent values. Symbolically, the references $r$ are updated with distinct fresh symbolic integers $\kappa$ if $r \in \refs_\ints$, distinct fresh method names if $r\in\refs_{\theta_1\to\theta_2}$, or to arbitrary reference names if $r\in\refs_{\refs_\theta}$.

%%% Local Variables:
%%% mode: latex
%%% TeX-master: "paper"
%%% End:

\shorten{\newcommand\code[1]{\sharp(#1)}

\section{Definability}\label{app:def}

In this section we show that every trace $\tau$ in the semantics of a library $L$ has a corresponding good client that realises the same trace in its semantics. 

Let $L$ be a library with public names $\pub$ and abstract names $\abs$.
Given a trace  $\tau$  produced by $L$, with  $\pub'$ and $\abs'$  respectively the public and abstract names introduced in $\tau$, we set:
\begin{align*}
  {\cal N} &= \pub \cup \pub' \cup \abs \cup \abs'\\
  \Theta_v &= \{\theta \mid \exists  m\in{\cal N}.\ m:\theta'\land \theta \text{ a syntactic subtype of }\theta'\} \\
  \Theta_m &= \{\theta\in\Theta \mid \theta \text{ a method type}\}
\end{align*}
Note that the above sets are finite, since $\tau,\pub,\abs$ are finite. We assume a fixed enumeration of ${\cal N}=\{m_1,m_2,\cdots,m_n\}$. Moreover, for each type $\theta$, we let $\mathbf{defval}_\theta$ be a default value, and $\mathbf{diverge}_\theta$ a term that on evaluation diverges by infinite recursion.
We then construct a client $C_{\tau,\pub,\abs}$ as in Figure~\ref{fig:defin}.
\begin{figure}[t]
\begin{center}
  \begin{lstlisting}[escapeinside={(*}{*)},basicstyle=\small]
  global cnt := 0
  global meth := 0
  global ref(*$_i$*) := (*$m_i$*)          # (*for each $m_i\in\pub$*)
  global ref(*$_i$*) := defval      # (*for each $m_i\in\pub'$*)
  global val(*$_\theta$*) := defval      # (*for each $\theta\in\Theta_v$*)
  public (*$m_i$*) = lambdax.            # (*for each $m_i\in\abs$*)
      cnt++; meth:=i; val(*$_{\theta_1}$*):=x; oracle()
  (*$m_i$*) = lambdax.                   # (*for each $m_i\in\abs'$*)
      cnt++; meth:=i; val(*$_{\theta_1}$*):=x; oracle()
  oracle = lambda().
      match (!cnt) with    # (*number of P-moves played so far (max $|\tau|/2$)*)
      | i -> 
        # (*if $i>0$ and $i$-th P-move of $\tau$ is $\mathtt{cr}\, m_j(v)$, with $m_j:\theta_1\to\theta_2$, then*)
        # (*- if $\mathtt{cr}=\mathtt{ret}$ then $d = 0$ and $\theta=\theta_2$*)
        # (*- if $\mathtt{cr}=\mathtt{call}$ then $d = j$  and $\theta=\theta_1$*)
        # (*diverge if the last P-move played is different from $\mathtt{cr}\, m_j(v)$*)
        if not (!meth = d and !val(*$_\theta$*) (*$\overset{\land}{=}_\theta$*) v) then diverge
        else for (*$m_i$*) in (*$fresh$*)(!val(*$_\theta$*)) do ref(*$_i$*) := (*$m_i$*)
        # (*if $(i+1)$-th O-move of $\tau$ is $\mathtt{cr}'\, m_k(u)$, with $m_k:\theta_1\to\theta_2$, then*)
        # (*- if $\mathtt{cr}'=\mathtt{ret}$ then $c = 0$*)
        # (*- if $\mathtt{cr}'=\mathtt{call}$ then $c=k$*)
        if c then let x = (!ref(*$_k$*))u in   # (*call $m_k(u)$*)
                  cnt++; meth:=0; val(*$_{\theta_2}$*):=x; oracle(); !val(*$_{\theta_2}$*)
        else val(*$_{\theta_2}$*):=u                   # (*return $u$*)
  main = oracle()
\end{lstlisting}
\end{center}\caption{The client $C_{\tau,\pub,\abs}$.}\label{fig:defin}
\end{figure}

\noindent
The code is structured as follows.
\begin{enumerate}
\item We start off by defining global references:
\begin{itemize}
\item $cnt$ counts the number of $P$ (Library) moves played so far;
\item $meth$ stores an index that records the move made by P: if the move was a return then $meth$ stores 0; if it was call to $m_i$ then $meth$ stores $i$;
\item each $ref_i$ will store the method $m_i\in\pub\cup\pub'$, either since the beginning (if $m_i\in\pub$), or once P plays it (if $m_i\in\pub'$);
\item each $val_\theta$ will be used for storing the value played by P in their last move. \end{itemize}
In the latter case above, there is a light abuse of syntax as $\theta$ can be a product type, of which \holib\ does not have references. But we can in fact simulate references of arbitrary type by several \holib\ references.
\item For each $m_i : \theta_1 \to \theta_2 \in \abs$, we define a public method $m_i$ that simulates the behaviour of O whenever $m_i$ is called in $\tau$:
  \begin{itemize}
  \item it starts by increasing $cnt$, as a call to $m_i$ corresponds to a P-move being played;
    \item it continues by storing $i$ and $x$ in $meth$ and $val_{\theta_1}$ respectively;
  \item it calls the private method $oracle$, which is tasked with simulating the rest of $\tau$ and storing the value that $m_i$ will return in $val_{\theta_2}$;
  \item it returns the value in $val_{\theta_2}$.
  \end{itemize}
\item For each $m_i:\theta_1\to\theta_2\in\abs'$ we produce a method just like above, but keep it private (for the time being).
\item The method $oracle$ performs the bulk of the computations, by checking that the last move played by P was the expected one and selecting the next move to play (and playing it if is a call).
  \begin{itemize}
  \item The oracle is called after each P-move is played, so it starts with increasing $cnt$.
  \item It then performs a case analysis on the value of $cnt$, which  above we  denote collectively by assuming the value is $i$ -- this notation hides the fact that we have one case for each of the finitely many values of $i$.

    For each such $i$,
    the oracle first checks if the previous P-move (if there was one), was the expected one. If the move was a call, it checks whether the called method was the expected one (via an appropriate value of $d$), and also whether the value was the expected one. Value comparisons ($\overset{\land}{=}_\theta$) only compare the integer components of $\theta$, since we cannot compare method names.
      If this check is successful, the oracle extracts from $u$ any method names played fresh by P and stores them in the corresponding $ref_i$.

      Next, the oracle prepares the next move. If, for the given $i$, the next move is a call, then the oracle issues the call, stores the return value of that call, increases $cnt$  and recurs to itself -- when the issued call returns, it would be through a P-move. If, on the other hand, the next move is a return, the oracle simply stores the value to be returned in the respective $val$ reference -- this would allow to the respective $m_i$ to return that value.
    \end{itemize}
\item The $\mathbf{main}$ method simply calls the oracle.
\end{enumerate}

We can then show the following (proof provided in full version~\cite{FULL}). 
For any library $L$ and $(\tau,\rho)\in\sem{L}$, $C_\tau$ is  such that $(\tau,\rho')\in\sem{C_\tau}$ for some $\rho'$.

% \begin{proof}
% Given a library $L$ and trace produced $\tau$, we construct client $C_{\tau}$. Since $C_\tau$ has a main method, we begin from a proponent configuration $(\mathtt{oracle}(),[],R,\abs,\pub)_p$. Since the library cannot return without being called first, we know the next move is a call, so $\tau$ is of the form $\mathtt{call}(m,v) \tau'$. Thus, we have the following transitions
% \[([],\mathtt{oracle}(),R,\abs,\pub)_p \twoheadrightarrow ([],E[m v],R,\abs,\pub)_p \to ((E,m)::[],R,\abs',\pub)_o \]
% From this point, if $\tau'$ is empty, we have shown that $\tau$ can be produced by $C_\tau$. If $\tau'$ is not empty, we have a trace $\tau$ with suffix $\tau'$ and prefix $\mathtt{call}(m,v)$. By Lemma~\ref{lem:clientconfigs}, we know that $\tau'$ can be produced by any configuration in $\mathtt{Config}_{\tau'}$. Since $((E,m)::[],R,\abs',\pub)_o \in \mathtt{Config}_{\tau'}$, we know that $((E,m)::[],R,\abs',\pub)_o$ is able to produce $\tau'$. We thus have that $C_\tau$ can produce $\tau$.
% \end{proof}

%%% Local Variables:
%%% mode: latex
%%% TeX-master: "paper"
%%% End:

}{\section{Soundness and Completeness}\label{apx:SC}

We prove here that the trace semantics for libraries is sound and complete: for any error that can be reached in the trace semantics there is a client such that linking the library with the client reaches the same value/error. And viceversa. In the following sections, we prove compositionality of our modified trace semantics. We use a bisimulation argument similar to~\cite{DBLP:journals/corr/MurawskiT16a}.

\subsection{Semantic Composition}

We start by defining a notion of composition that combines the traces produced by two configurations. These are supposed to correspond to a library and a client, but for now we will only require that the configurations satisfy a set of compatibility conditions.

We say configurations $\rho$ and $\rho'$ of \emph{opposite polarity} (one is $p$ if the other is $o$) are compatible ($\rho \asymp \rho'$) if:
\begin{itemize}
\item their stores are disjoint: $\refs(\rho) \cap \refs(\rho') = \varnothing$
\item $\rho$ closes and is closed by $\rho'$: $\pub = \abs'$ and $\pub' = \abs$
\item undisclosed names of $\rho$ do not occur in $\rho'$ and vice versa: $(\meths(\rho) \setminus (\abs \cup \pub)) \cap \meths(\rho') = \varnothing$
\item their evaluation stacks are compatible, written $\eval \asymp \eval'$, which means:
\begin{itemize}
\item $\eval = \eval'=\varepsilon$; or
\item $\eval = (m,l)::\eval_1$ and $\eval' = (m,E)::\eval_1'$, and $\eval_1 \asymp \eval_1'$; or
\item $\eval = (m,E)::\eval_1$ and $\eval' = (m,l)::\eval_1'$, and $\eval_1 \asymp \eval_1'$.
\end{itemize}
\end{itemize}
Note that compatibility of evaluation stacks expects that compatible configurations are always of opposite polarity. This reflects the fact that we compose libraries with closing clients.

\cutout{may seem questionable at first since general compositionality allows a client to have an environment of its own, independent of the library, and thus requires games between the composed configurations--as two opponent configurations--and an external client. However, since we only aim to model check libraries, we assume libraries and clients close each other, meaning that the client captures the whole environment and would include any external interactions as internal moves.}

With these definitions, we follow by defining different notions of composition.

Let $\rho_1,\rho_2,\rho_1',\rho_2'$ be game configurations. The following rules define the semantic composition of two configurations.
\begin{align*}
{\begin{prooftree}
\Hypo{ \rho_1 \to' \rho_1' }
\Hypo{ \rho_2' = \rho_2 }
\Infer2[ $\textsc{Int}_L$ ]
{ \rho_1 \semcomp \rho_2 \to' \rho_1' \semcomp \rho_2' }
\end{prooftree}}
\quad
{\begin{prooftree}
\Hypo{ \rho_2 \to' \rho_2' }
\Hypo{ \rho_1' = \rho_1 }
\Infer2[ $\textsc{Int}_C$ ]
{ \rho_1 \semcomp \rho_2 \to' \rho_1' \semcomp \rho_2' }
\end{prooftree}}
\end{align*}
\begin{align*}
{\begin{prooftree}
\Hypo{ \rho_1 \xrightarrow{\mathtt{call }(m,v)}{\!\!}^\prime \rho_1' }
\Hypo{ \rho_2 \xrightarrow{\mathtt{call }(m,v)}{\!\!}^\prime \rho_2' }
\Infer2[ $\textsc{Call}$ ]
{ \rho_1 \semcomp \rho_2 \to' \rho_1' \semcomp \rho_2' }
\end{prooftree}}
\end{align*}
\begin{align*}
{\begin{prooftree}
\Hypo{ \rho_1 \xrightarrow{\mathtt{ret }(m,v)}{\!\!}^\prime \rho_1' }
\Hypo{ \rho_2 \xrightarrow{\mathtt{ret }(m,v)}{\!\!}^\prime \rho_2' }
\Infer2[ $\textsc{Ret}$ ]
{ \rho_1 \semcomp \rho_2 \to' \rho_1' \semcomp \rho_2' }
\end{prooftree}}
\end{align*}

\subsection{Composite Semantics and Internal Composition}

We now introduce the notion of composing game configurations \boldemph{internally}, which occurs when merging two compatible game configurations into a single composite semantics configuration. We first refine the operational semantics and produce a \boldemph{composite semantics}. This is necessary for our compositionality argument since there is an asymmetry between the call counters of the opponent and proponent configurations. Proponent configurations count calls internally while opponent configurations have no internal counters, and thus only count calls when playing moves. This requires that we keep track of two pairs of counters, one for each component, which may change at different rates.

With this in mind, to define the composite semantics, we extend the term configurations to obtain tuples of the following form:
\[ (M,R_1,R_2,S,k_1,k_2,l_1,l_2) \text{ for which we shall write } (M,\vec{R},S,\vec{k},\vec{l}) \]
where $R_1$ and $R_2$ are the library and client methods respectively, such that $dom(R_1) \cap dom(R_2) = \varnothing$, $S$ is the combined store, and $k_1,l_1$ and $k_2,l_2$ are counters managed by the library and client.
All operators tagged with $i$ will be operating on the $i$th component; e.g. $\vec{R}[m \mapsto M]_i$ states that $R_i[m \mapsto M]$ in $\vec{R}$. We also extend $M$ by tagging all method names (written $m^i$) as well as all lambda-abstractions (written $\lambda^i$) with $i\in\{1,2\}$ to show whether they are being called from the library (1) or the client (2). We write $M^i$ to be the term $M$ with all its methods and lambdas tagged with $i$. Evaluation contexts are also extended to mark methods which are being called from the opposite polarity:
\[ E ::= \dots \mid \evalbox{E}^i \mid \evaltag{E}{i,l}\]
Intuitively, $i$ is the component that is currently at a proponent configuration in the equivalent game semantics, while $l$ in $\evaltag{E}{i,l}$ is the opponent counter for component $3-i$. This will be used particularly when evaluating a method call $m^i v$ when $m \nin dom(R_i)$. Applying these changes, we define the semantics for composite terms ($\compsemto$).
 \begin{flalign*}
 &(E[\assert(i)],\vec{R},S,\vec{k},\vec{l}) \compsemto (E[()],\vec{R},S,\vec{k},\vec{l}) \quad (i \neq 0)\\
 &(E[!r],\vec{R},S,\vec{k},\vec{l}) \compsemto  (E[S(r)],\vec{R},S,\vec{k},\vec{l})\\
 &(E[r:=v],\vec{R},S,\vec{k},\vec{l}) \compsemto (E[()],\vec{R},S[r\mapsto v],\vec{k},\vec{l})\\
 &(E[\pi_j\pair{v_1,v_2}],\vec{R},S,\vec{k},\vec{l}) \compsemto (E[v_j],\vec{R},S,\vec{k},\vec{l})\\
 &(E[i_1 \oplus i_2],\vec{R},S,\vec{k},\vec{l}) \compsemto (E[i],\vec{R},S,\vec{k},\vec{l})\quad (i=i_1\oplus i_2)\\
 &(E[\lambda^i x.M],\vec{R},S,\vec{k},\vec{l}) \compsemto (E[m],\vec{R}[m\mapsto \lambda x.M]_i,S,\vec{k},\vec{l})\quad (m \not\in dom(\vec{R}))\\
 &(E[\ifm{i}{M_1}{M_0}],\vec{R},S,\vec{k},\vec{l}) \compsemto (E[M_j],\vec{R},S,\vec{k},\vec{l})\quad (j = 1 \text{ iff } i \neq 0)\\
 &(E[\letm{x=v}{M}],\vec{R},S,\vec{k},\vec{l}) \compsemto (E[M\{v/x\}],\vec{R},S,\vec{k},\vec{l})\\
 &(E[\letrec{f=\lambda^i x.M}{M'}],\vec{R},S,\vec{k},\vec{l}) \compsemto (E[M'\{m/f\}],\vec{R}[m \mapsto \lambda x.M\{m/f\}]_i,S,\vec{k},\vec{l})\\
 &(E[m^i v],\vec{R},S,\vec{k},\vec{l}) \compsemto (E[\evalbox{M\{v/y\}^i}^i],\vec{R},S,\vec{k}+_i 1,\vec{l})\quad(R_i(m) = \lambda y.M)\\
 &(E[m^i v],\vec{R},S,\vec{k},\vec{l}) \compsemto (E[\evaltag{m^{3-i}v}{i,l+_{3-i}1}],\vec{R},S,\vec{k},\vec{l}[l_i \mapsto 0])\quad(R_{3-i}(m) = \lambda y.M)\\
 &(E[\evalbox{v}^i],\vec{R},S,\vec{k}+_i 1,\vec{l}) \compsemto (E[v^i],\vec{R},S,\vec{k},\vec{l})\\
 &(E[\evaltag{v}{i,l}],\vec{R},S,\vec{k},\vec{l}) \compsemto (E[v^i],\vec{R},S,\vec{k},\vec{l}[l_{3-i} \mapsto l,l_i \mapsto last(E)])\\
 &\qquad \text{if $last(E)$ is defined,}\\
 &\qquad \text{and $last(E) = \hat{l}$ if $E = E_1[\evaltag{E_2}{j,\hat{l}}]$ provided $E_2$ has no tags $\evaltag{\bullet}{j',\hat{l}'}$}
 \end{flalign*}

We continue by defining the \boldemph{internal composition} of compatible configurations $\rho_1 \asymp \rho_2$. We define the internal composition $\rho_1 \confcomp \rho_2$ to be a configuration in our new composite semantics by pattern matching on the configuration polarity and evaluation stacks according to the following rules. For clarity, we annotate opponent and proponent configurations with $o$ and $p$ respectively.

\paragraph*{Initial Configuration:}
\begin{align*}
\rho_1 &= ([],-,R_1,S_1,\pub_1,\abs_1,0,0)_o\\
\rho_2 &= ([],M_0,R_2,S_2,\pub_2,\abs_2,0,-)_p\\
\rho_1 \confcomp \rho_2 &= (\evaltag{\bullet}{1,0}[M_0^2],R_1,R_2,S_1 \uplus S_2,0,0,0,0)
\end{align*}

\paragraph*{Interim Configuration (case OP):}
\begin{align*}
\rho_1 &= (\eval_1,-,R_1,S_1,\pub_1,\abs_1,k_1,l_1)_o\\
\rho_2 &= (\eval_2,M,R_2,S_2,\pub_2,\abs_2,k_2,-)_p\\
\eval_1 &= (m,E)::\eval_1' \quad \eval_2 = (m,l_2)::\eval_2'\\
\rho_1 \confcomp \rho_2 &= ((\eval_1' \confcomp \eval_2')[E^1[\evaltag{M^2}{1,l_2}]],R_1,R_2, S_1 \uplus S_2, k_1,k_2,l_1,l_2)
\end{align*}

\paragraph*{Interim Configuration (case PO):}
\begin{align*}
\rho_1 &= (\eval_1,M,R_1,S_1,\pub_1,\abs_1,k_1,-)_p\\
\rho_2 &= (\eval_2,-,R_2,S_2,\pub_2,\abs_2,k_2,l_2)_o\\
\eval_1 &= (m,l_1)::\eval_1' \quad \eval_2 = (m,E)::\eval_2'\\
\rho_1 \confcomp \rho_2 &= ((\eval_1' \confcomp \eval_2')[E^2[\evaltag{M^1}{2,l_1}]],R_1,R_2,S_1 \uplus S_2, k_1,k_2,l_1,l_2)
\end{align*}
where $\eval_1' \confcomp \eval_2'$ is a single evaluation context resulting from the composition of compatible stacks $\eval_1'$ and $\eval_2'$, which we define as follows:
\begin{align*}
\varepsilon \confcomp \varepsilon &= \bullet\\
((m',E)::\eval_1'') \confcomp ((m',l)::\eval_2'') &= (\eval_1'' \confcomp \eval_2'')[E^1[\evaltag{\bullet}{1,l}]]\\
((m',l)::\eval_1'') \confcomp ((m',E)::\eval_2'') &= (\eval_1'' \confcomp \eval_2'')[E^2[\evaltag{\bullet}{2,l}]]
\end{align*}
Notice that there is only one case for initial configurations, and that is because the game must start from an opponent-proponent configuration where stacks are empty.

\subsection{Bisimilarity of Semantic and Internal Composition}
 We begin by defining bisimilarity for the semantic and internal composition. A set $\mathcal{R}$ with elements of the form $(\rho_1,\rho_2)$, {where $\rho_1$ is a configuration of the form $\rho_1'\semcomp\rho_1''$ and $\rho_2$ is from the composite semantics}, is a \boldemph{bisimulation} if for all $(\rho_1,\rho_2) \in R$:
\begin{itemize}
\item if $\rho_1 \to' \rho_1'$ then $\rho_2 \compsemto^* \rho_2'$ and $(\rho_1',\rho_2') \in \mathcal{R}$;
\item if $\rho_2 \compsemto \rho_2'$ then $\rho_1 \to'^* \rho_1'$ and $(\rho_1',\rho_2') \in \mathcal{R}$.
\end{itemize}
We say that two game configurations $\rho,\rho'$ are \emph{bisimilar}, and write $\rho \sim \rho'$, if there is a bisimulation $\mathcal{R}$ such that $\rho \mathcal{R} \rho'$.

Lemma \ref{bisim} states that, given game configurations, it is possible to obtain the composite semantics ($\compsemto$) from the semantic composition of the corresponding compatible configurations, and vice versa.

{\lemma{\label{bisim}}{Given game configurations $\rho \asymp \rho'$, it is the case that $(\rho \semcomp \rho') \sim (\rho \confcomp \rho')$.}}

\begin{proof}
We want to show that $\mathcal{R} = \{(\rho_1 \semcomp \rho_2, \rho_1 \confcomp \rho_2) \mid \rho_1 \asymp \rho_2\}$ is a bisimulation. Suppose $(\rho_1 \semcomp \rho_2, \rho_1 \confcomp \rho_2) \in \mathcal{R}$. We begin with case analysis on the transitions available to the semantic composite. If $(\rho_1 \semcomp \rho_2) \to' (\rho_1' \semcomp \rho_2')$, then $\rho_1' \asymp \rho_2'$. Now, by cases of the transitions, we prove that composite semantics can be obtained from the semantic composition.
\begin{enumerate}
\item If $(\rho_1 \semcomp \rho_2) \to' (\rho_1' \semcomp \rho_2')$ is an ($\textsc{Int}_L$) move, then we have internal moves in the execution of $\rho_1$ up to $\rho_1'$. Since the composite semantics is concrete and, by construction, equivalent to operational semantics when no methods of opposite polarity are called, we can see that $(\rho_1 \confcomp \rho_2) \compsemto (\rho_1' \confcomp \rho_2)$.

\item If $(\rho_1 \semcomp \rho_2) \to' (\rho_1' \semcomp \rho_2')$ is a ($\textsc{Call}$) move, then we have that $\rho_1 \xrightarrow{\mathtt{call}(m,v)}{\!\!}^\prime \rho_1'$ and $\rho_2 \xrightarrow{\mathtt{call}(m,v)}{\!\!}^\prime \rho_2'$. We thus have two cases: (1) $m$ is defined in $R_1$ and (2) it is in $R_2$. In case (1), we have the following semantics for $\rho_1$ and $\rho_2$ where the evaluation stacks are not equal:
\begin{align*}
&((m',E')::\eval_1,-,R_1,S_1,\pub_1,\abs_1,k_1,l_1)_o\\
&\quad \xrightarrow{\mathtt{call}(m,v)}{\!\!}^\prime((m,l_1 +1)::(m',E')::\eval_1,mv,R_1,S_1,\pub_1,\abs_1',k_1,-)_p\\
&((m',l_2)::\eval_2, E[mv], R_2, S_2, \pub_2, \abs_2, k_2, -)_p\\
&\quad \xrightarrow{\mathtt{call}(m,v)}{\!\!}^\prime((m,E)::(m',l_2)::\eval_2,-,R_2,S_2,\pub_2',\abs_2,k_2,l_0)_o
\end{align*}
We thus have:
\begin{align*}
&\rho_1 \confcomp \rho_2 = ((\eval_1 \confcomp \eval_2)[E'^1[\evaltag{E^2[m^2v]}{1,l_2}]],\vec{R},S_1 \cup S_2,\vec{k},\vec{l}) \\
&\rho_1' \confcomp \rho_2' = ((\eval_1 \confcomp \eval_2)[E'^1[\evaltag{E^2[\evaltag{m^1v}{2,l_1+1}]}{1,l_2}]],\\
&\qquad \qquad \quad \vec{R},S_1 \cup S_2,\vec{k},\vec{l}[l_2\mapsto 0]+_1 1)
\end{align*}
From the composite semantics evaluating $\rho_1 \confcomp \rho_2$ we have:
\begin{align*}
&((\eval_1 \confcomp \eval_2)[E'^1[\evaltag{E^2[m^2v]}{1,l_2}]],\vec{R},S_1 \cup S_2,\vec{k},\vec{l})\\
&\compsemto ((\eval_1 \confcomp \eval_2)[E'^1[\evaltag{E^2[\evaltag{m^1 \hat{v}}{2,l_1+1}]}{1,l_2}]],\\
&\qquad \qquad \vec{R},S_1 \cup S_2,\vec{k},\vec{l}[l_2\mapsto 0]+_{1}1)
\end{align*}
Since $v = \hat{v}$ by determinism of the operational semantics, we have that $(\rho_1 \confcomp \rho_2) \compsemto (\rho_1' \confcomp \rho_2')$. In addition, we can observe that the case for equal evaluation stacks is proven by substituting the initial stacks with equal ones, which results in an empty evaluation context. Similarly, the dual case (2), where $m$ is defined in $R_1$, is identical but with polarities swapped--i.e. shown by the polar complement of $(\rho_1 \confcomp \rho_2) \compsemto (\rho_1' \confcomp \rho_2')$.

\item If $(\rho_1 \semcomp \rho_2) \to' (\rho_1' \semcomp \rho_2')$ is a ($\textsc{Ret}$) move, then we have that $\rho_1 \xrightarrow{\mathtt{ret}(m,v)}{\!\!}^\prime \rho_1'$ and $\rho_2 \xrightarrow{\mathtt{ret}(m,v)}{\!\!}^\prime \rho_2'$. As with the \textsc{Call} case,  if $m \in dom(R_2)$ and stacks are not equal, we have:
\begin{align*}
&((m,E)::\eval_1, - , R_1,S_1, \pub_1, \abs_1, k_1, l_1)_o\\
&\quad \xrightarrow{\mathtt{ret}(m,v)}{\!\!}^\prime(\eval_1, E[v],R_1,S_1,\pub_1,\abs_1', k_1, -)_p\\
&((m,l_2)::\eval_2,v,R_2,S_2,\pub_2,\abs_2,k_2,-)_p\\
&\quad \xrightarrow{\mathtt{ret}(m,v)}{\!\!}^\prime(\eval_2, - , R_2,S_2, \pub_2', \abs_2, k_2, l_2)_o
\end{align*}
Here, we have two cases: $\eval_1 = \eval_2$, and otherwise. We start with the case where $\eval_1 \neq \eval_2$, since the opposite case is a simpler version of it. Again, we have the following composite configurations:
\begin{align*}
&\rho_1 \confcomp \rho_2 = ((\eval_1 \confcomp \eval_2)[E^1[\evaltag{v^2}{1,l_2}]],\vec{R},S_1 \cup S_2,\vec{k},\vec{l})\\
&\rho_1' \confcomp \rho_2' = ((\eval_1' \confcomp \eval_2')[E'^2[\evaltag{E^1[v^1]}{2,l_1'}]],\\
&\qquad \qquad \quad \vec{R},S_1 \cup S_2,\vec{k},l_1',l_2)
\end{align*}
where $\eval_1 = (m',l_1')::\eval_1'$ and $\eval_2 = (m',E')::\eval_2$.
	
Now, from the composite semantics, we have:
\begin{align*}
&((\eval_1 \confcomp \eval_2)[E^1[\evaltag{v^2}{1,l_2}]],\vec{R},S_1 \cup S_2,\vec{k},\vec{l})\\
&\compsemto ((\eval_1 \confcomp \eval_2)[E^1[\hat{v}^1]],\vec{R},S_1 \cup S_2,\vec{k},last((\eval_1 \confcomp \eval_2)[E^1[\bullet]]),l_2)\\
&~= ((\eval_1' \confcomp \eval_2')[E'^2[\evaltag{E^1[\hat{v}^1]}{2,l_1'}]],\vec{R},S_1 \cup S_2,\vec{k},l_1',l_2)
\end{align*}

We can observe that $last(E) = l_1'$ since $E$ comes directly from the evaluation stack and is, thus, untagged, and the top-most counter is $l_1'$ since 
\[(\eval_1' \confcomp \eval_2')[E'^2[\evaltag{E^1[\bullet]}{2,l_1'}]] = (\eval_1 \confcomp \eval_2)[E^1[\bullet]]\]

Finally, we have that $k_2 = k_2'$ when returning a value since, from Lemma~\ref{kbounds}, $k$ must always decrease back to its original value after evaluating a method call.

We thus have $(\rho_1 \confcomp \rho_2) \compsemto (\rho_1' \confcomp \rho_2')$. As previously, the case for empty stacks is a simpler version of this, while the dual case (2) is the polar complement of the configurations.
\end{enumerate}
	
Having shown that external composition produces composite semantics transitions, we continue with the other direction of the argument, which aims to show that the external composition can be produced from composite semantics transitions. We now derive the corresponding semantic compositions by case analysis on the composite semantics rules.
\begin{enumerate}
\item If we have an untagged transition, or one where the redex involves no names of opposite polarity being called, then we have an exact correspondence with internal moves, since the composite semantics are identical to the operational semantics on closed terms.

\item If the transition involves a method called from an opposite polarity, we have a transition of the form
\[(E[m^i v],\dots,\vec{l}) \compsemto (E[\evaltag{m^{3-i} v}{i,l_{3-i}+1}],\dots,\vec{l}[l_i \mapsto 0]+_{3-i}1)\]
which corresponds to evaluating the semantics on an initial configuration $\rho_1 \confcomp \rho_2$ with the following cases:
\begin{enumerate}
\item for an OP configuration, we have the following: 
\[\rho_1 = (\eval_1, -, R_1, S_1, k_1, l_1)_o\] 
\[\rho_2 = (\eval_2, E[m v], R_2, S_2, k_2, -)_p\]
where $\eval_1 = (m',E')::\eval_1'$ and $\eval_2 = (m',l_2)::\eval_2'$. Let us set $E[m^i v] = (\eval_1' \confcomp \eval_2')[E'^1[\evaltag{M^2}{1,l_2}]]$ and $M^2 = E''[m^i v]$, where $m \nin R_2$, $i=2$, and $E''$ is untagged. We therefore have:
\begin{align*}
&((\eval_1' \confcomp \eval_2')[E'^1[\evaltag{M^2}{1,l_2}]],\vec{R},S_1 \cup S_2,\vec{k},\vec{l})\\
&\compsemto (E[\evaltag{m^{1} v}{2,l_1+1}],\vec{R},S_1 \cup S_2,\vec{k},\vec{l}[l_2 \mapsto 0] +_1 1)
\end{align*}
We now want to show that semantically composing the configurations results in an equivalent transition $\rho_1 \semcomp \rho_2 \to' \rho_1' \semcomp \rho_2'$. Since this is a $\textsc{Call}$ move, we know that $\rho_1 \xrightarrow{\mathtt{call}(m,v)}{\!\!}^\prime\rho_1'$ and $\rho_2 \xrightarrow{\mathtt{call}(m,v)}{\!\!}^\prime\rho_2'$. Evaluating those transitions, we have that
\[\rho_1' = ((m,l_1+1)::\eval_1,m v, \dots, k_1, -)_o\]
\[\rho_2' = ((m,E'')::\eval_2, -, \dots, k_2, 0)_p\]
which, when syntactically composed, form the configuration 
\[((\eval_1 \confcomp \eval_2)[E''^2[\evaltag{(m v)^1}{2,l_1+1}]], \vec{R}, S_1 \cup S_2, \vec{k}, \vec{l}[l_2 \mapsto 0] +_1 1)\]
We can observe that the resulting configurations are equivalent since $E'' = E''^2$, which follows from $E''[m^i v] = M^2$. Additionally, since
\[(\eval_1' \confcomp \eval_2')[E'^1[\evaltag{E''^2[\bullet]}{1,l_2}]] = (\eval_1 \confcomp \eval_2)[E''^2[\evaltag{\bullet}{2,l_1+1}]]\]
it suffices to show $(m v)^1 = m^1 v$, particularly that $v = v^1$. Now, since the composite semantics ensures that $v$ will be tagged with $1$ when called from a method $m^1$, as it reduces to $M\{v/y\}^1$, we have that $v = v^1$, meaning that the transitions are equal.
			
\item for a PO configuration, the polar complement of case (a) suffices.
			
\item for an initial configuration OP, we have a simpler version of case (a) where the evaluation stacks are equal, resulting in an empty evaluation context $\eval_1' \confcomp \eval_2' = \bullet$.
\end{enumerate}
		
\item If the transition involves a tagged value and is of the form
\begin{align*}
&(E[\evaltag{v}{i,l}],\vec{R},S_1 \cup S_2,\vec{k},\vec{l}) \\
&\compsemto (E[v^{i}],\vec{R},S_1 \cup S_2,\vec{k},\vec{l}[l_{3-i}\mapsto l, l_i \mapsto last(E)])
\end{align*}
then we want to show an equivalence to a \textsc{Ret} move in the semantic composite. As with case (2), we start by defining this transition as the syntactic composite transition $(\rho_1 \confcomp \rho_2) \compsemto (\rho_1' \confcomp \rho_2')$. Then, by case analysis on $\rho_1 \confcomp \rho_2$:
\begin{enumerate}
\item for an OP configuration, we have the following:
\[\rho_1 = (\eval_1, -, R_1,S_1, k_1, l_1)_o\] 
\[\rho_2 = (\eval_2, v, R_2,S_2, k_2, -)_p\]
where $\eval_1 = (m,E')::\eval_1'$ and $\eval_2 = (m,l_2)::\eval_2'$. Let $E[v] = (\eval_1' \confcomp \eval_2')[E'^1[\evaltag{v^2}{1,l_2}]]$. We thus have:
\begin{align*}
&(E[\evaltag{v^2}{1,l_2}],\vec{R},S_1 \cup S_2,\vec{k},\vec{l})\compsemto(E[v^1],\vec{R},S_1 \cup S_2,\vec{k},last(E),l_2)
\end{align*}
We then show that semantic composition produces an equivalent transition $\rho_1 \semcomp \rho_2 \to' \rho_1' \semcomp \rho_2'$. Given we have a \textsc{Ret} move, we know that $\rho_1 \xrightarrow{\mathtt{ret}(m,v)}{\!\!}^\prime \rho_1'$ and $\rho_2 \xrightarrow{\mathtt{ret}(m,v)}{\!\!}^\prime \rho_2'$, such that:
\[\rho_1' = (\eval_1', E'[v], \dots, k_1, -)_p\]
\[\rho_2' = (\eval_2', -, \dots, k_2, l_2)_o\]
where $\eval_1' = (m',l_1')::\eval_1''$ and $\eval_2' = (m',E')::\eval_2''$. Internally composing these resulting configurations, we have:
\[((\eval_1'' \confcomp \eval_2'')[E''[\evaltag{E'^1[v^1]}{2,l_1'}]], \vec{R'}, S_1 \cup S_2, \vec{k}, l_1',l_2)\]
Since $(\eval_1'' \confcomp \eval_2'')[E''[\evaltag{\bullet}{2,l_1'}]] = (\eval_1' \confcomp \eval_2')[\bullet]$, we have that $(\eval_1' \confcomp \eval_2')[E'^1[v^1]]$, from which we have $(\eval_1'' \confcomp \eval_2'')[E''[\evaltag{E'^1[v^1]}{2,l_1'}]] = E[v^1]$, and that $last(E) = l_1'$ since $E_1'$ is untagged. Thus, the transition produces the composition.
			
\item for a PO configuration, we have the polar complement of (a) as previously.
			
\item for an initial OP configuration, we again have a simplification of (a), where the evaluation stacks are equal and the resulting evaluation context is empty.
\end{enumerate}		
\end{enumerate}
	
With this, we are done showing the equivalence of transitions. Lastly, we can observe that $\rho$ is final iff $\rho'$ is final since they are both leaf nodes generated by equivalent terminal rules. We therefore have $(\rho \semcomp \rho') \sim (\rho \confcomp \rho')$.
\end{proof}
	
\subsection{Syntactic Composition and Compositionality}

To prove compositionality of the modified trace semantics, we want to show that syntactic composition can be obtained from semantic counterpart and vice versa. We have bisimilarity between semantic and internal composition, we only need to show that internal composition is related to syntactic composition under some notion of equivalence. 

\paragraph*{\boldemph{Lemma~\ref{simsim}}}
For any library $L$ and compatible good client $C$, $\sem{L \comp C}$ fails if and only if there exist $(\tau_1,\rho_1) \in \sem{L}$ and $(\tau_2,\rho_2) \in \sem{C}$ such that $\tau_1=\tau_2$ and $\rho_1=(\eval,E[\assert(0)],\cdots)$.
\begin{proof}
We have a case for each direction.
\paragraph*{\bf $(1 \implies 2)$:}
\begin{enumerate}
	\item Consider $L\comp C$ that reaches $\chi$. 
	\item By inspection of the composite semantics, we have that $\sem{L} \confcomp \sem{C}$ reaches $\chi$.
	\item By bisimilarity (Lemma~\ref{bisim}) we have that $\sem{L} \semcomp \sem{C}$ reaches $\chi$.
	\item By definition of semantic composition, we know there are traces $\tau \in \sem{L}$ and $\tau^\bot \in \sem{C}$  such that $\sem{L} \xtwoheadrightarrow{\tau}' \chi$.
\end{enumerate}

\paragraph*{\bf $(2 \implies 1)$:}
\begin{enumerate}
	\item Consider traces $\tau \in \sem{L}$ and $\tau^\bot \in \sem{C}$  such that $\sem{L} \xtwoheadrightarrow{\tau}' \chi$.
	\item By definition of semantic composition we have that $\sem{L} \semcomp \sem{C}$ reaches $\chi$.
	\item By bisimilarity (Lemma~\ref{bisim}) we have that $\sem{L}\confcomp \sem{C}$ reaches $\chi$.
	\item By inspection of the composite semantics, we know $L\syncomp C$ reaches $\chi$. 
\end{enumerate}

\end{proof}

\newcommand\code[1]{\sharp(#1)}

\subsection{Definability}\label{app:def:full}

In this section we show that every trace $\tau$ in the semantics of a library $L$ has a corresponding good client that realises the same trace in its semantics. 

Let $L$ be a library with public names $\pub$ and abstract names $\abs$.
Given a trace  $\tau$  produced by $L$, with  $\pub'$ and $\abs'$  respectively the public and abstract names introduced in $\tau$, we set:
\begin{align*}
  {\cal N} &= \pub \cup \pub' \cup \abs \cup \abs'\\
  \Theta_v &= \{\theta \mid \exists  m\in{\cal N}.\ m:\theta'\land \theta \text{ a syntactic subtype of }\theta'\} \\
  \Theta_m &= \{\theta\in\Theta \mid \theta \text{ a method type}\}
\end{align*}
Note that the above sets are finite, since $\tau,\pub,\abs$ are finite. We assume a fixed enumeration of ${\cal N}=\{m_1,m_2,\cdots,m_n\}$. Moreover, for each type $\theta$, we let $\mathbf{defval}_\theta$ be a default value, and $\mathbf{diverge}_\theta$ a term that on evaluation diverges by infinite recursion.
We then construct a client $C_{\tau,\pub,\abs}$ as in Figure~\ref{fig:defin}.
\begin{figure}[t]
\begin{center}
  \begin{lstlisting}[escapeinside={(*}{*)},basicstyle=\small]
  global cnt := 0
  global meth := 0
  global ref(*$_i$*) := (*$m_i$*)          # (*for each $m_i\in\pub$*)
  global ref(*$_i$*) := defval      # (*for each $m_i\in\pub'$*)
  global val(*$_\theta$*) := defval      # (*for each $\theta\in\Theta_v$*)

  public (*$m_i$*) = lambdax.            # (*for each $m_i\in\abs$*)
      cnt++; meth:=i; val(*$_{\theta_1}$*):=x; oracle()

  (*$m_i$*) = lambdax.                   # (*for each $m_i\in\abs'$*)
      cnt++; meth:=i; val(*$_{\theta_1}$*):=x; oracle()

  oracle = lambda().
      match (!cnt) with    # (*number of P-moves played so far (max $|\tau|/2$)*)
      | i -> 
        # (*if $i>0$ and $i$-th P-move of $\tau$ is $\mathtt{cr}\, m_j(v)$, with $m_j:\theta_1\to\theta_2$, then*)
        # (*- if $\mathtt{cr}=\mathtt{ret}$ then $d = 0$ and $\theta=\theta_2$*)
        # (*- if $\mathtt{cr}=\mathtt{call}$ then $d = j$  and $\theta=\theta_1$*)
        # (*diverge if the last P-move played is different from $\mathtt{cr}\, m_j(v)$*)
        if not (!meth = d and !val(*$_\theta$*) (*$\overset{\land}{=}_\theta$*) v) then diverge
        else for (*$m_i$*) in (*$fresh$*)(!val(*$_\theta$*)) do ref(*$_i$*) := (*$m_i$*)
        
        # (*if $(i+1)$-th O-move of $\tau$ is $\mathtt{cr}'\, m_k(u)$, with $m_k:\theta_1\to\theta_2$, then*)
        # (*- if $\mathtt{cr}'=\mathtt{ret}$ then $c = 0$*)
        # (*- if $\mathtt{cr}'=\mathtt{call}$ then $c=k$*)
        if c then let x = (!ref(*$_k$*))u in   # (*call $m_k(u)$*)
                  cnt++; meth:=0; val(*$_{\theta_2}$*):=x; oracle(); !val(*$_{\theta_2}$*)
        else val(*$_{\theta_2}$*):=u                   # (*return $u$*)

  main = oracle()
\end{lstlisting}
\end{center}\caption{The client $C_{\tau,\pub,\abs}$.}\label{fig:defin}
\end{figure}

\noindent
The code is structured as follows.
\begin{enumerate}
\item We start off by defining global references:
\begin{itemize}
\item $cnt$ counts the number of $P$ (Library) moves played so far;
\item $meth$ stores an index that records the move made by P: if the move was a return then $meth$ stores 0; if it was call to $m_i$ then $meth$ stores $i$;
\item each $ref_i$ will store the method $m_i\in\pub\cup\pub'$, either since the beginning (if $m_i\in\pub$), or once P plays it (if $m_i\in\pub'$);
\item each $val_\theta$ will be used for storing the value played by P in their last move. \end{itemize}
In the latter case above, there is a light abuse of syntax as $\theta$ can be a product type, of which \holib\ does not have references. But we can in fact simulate references of arbitrary type by several \holib\ references.
\item For each $m_i : \theta_1 \to \theta_2 \in \abs$, we define a public method $m_i$ that simulates the behaviour of O whenever $m_i$ is called in $\tau$:
  \begin{itemize}
  \item it starts by increasing $cnt$, as a call to $m_i$ corresponds to a P-move being played;
    \item it continues by storing $i$ and $x$ in $meth$ and $val_{\theta_1}$ respectively;
  \item it calls the private method $oracle$, which is tasked with simulating the rest of $\tau$ and storing the value that $m_i$ will return in $val_{\theta_2}$;
  \item it returns the value in $val_{\theta_2}$.
  \end{itemize}
\item For each $m_i:\theta_1\to\theta_2\in\abs'$ we produce a method just like above, but keep it private (for the time being).
\item The method $oracle$ performs the bulk of the computations, by checking that the last move played by P was the expected one and selecting the next move to play (and playing it if is a call).
  \begin{itemize}
  \item The oracle is called after each P-move is played, so it starts with increasing $cnt$.
  \item It then performs a case analysis on the value of $cnt$, which  above we  denote collectively by assuming the value is $i$ -- this notation hides the fact that we have one case for each of the finitely many values of $i$.

    For each such $i$,
    the oracle first checks if the previous P-move (if there was one), was the expected one. If the move was a call, it checks whether the called method was the expected one (via an appropriate value of $d$), and also whether the value was the expected one. Value comparisons ($\overset{\land}{=}_\theta$) only compare the integer components of $\theta$, since we cannot compare method names.
      If this check is successful, the oracle extracts from $u$ any method names played fresh by P and stores them in the corresponding $ref_i$.

      Next, the oracle prepares the next move. If, for the given $i$, the next move is a call, then the oracle issues the call, stores the return value of that call, increases $cnt$  and recurs to itself -- when the issued call returns, it would be through a P-move. If, on the other hand, the next move is a return, the oracle simply stores the value to be returned in the respective $val$ reference -- this would allow to the respective $m_i$ to return that value.
    \end{itemize}
\item The $\mathbf{main}$ method simply calls the oracle.
\end{enumerate}

Let us begin with useful definitions. First, let us consider the game semantics for $\holib$ with all call counters removed since they do not affect computation. Let $L$ be a library with public names $\pub$ and abstract names $\abs$ that produces a trace $\tau$. Let $C_{\tau,\pub,\abs}$ be the client constructed from $\tau$, which we shall shorthand as $C_\tau$ assuming the correct name sets have been provided. Finally, let us annotate every move in $\tau$ with subscripts $O$ and $P$ for its polarity, starting from $O$ since libraries are always called first.

\begin{definition}[Client O-configurations]
Let library trace $\tau$ be of the form $\tau_{1} \tau_2$, where $\tau_1$ is the portion of $\tau$ that has been played so far. We define the set of opponent configurations $\mathtt{Conf}_{\tau_2}$ that play the remainder trace $\tau_2$ of trace $\tau$ to be
\[(\eval_{\tau_1}, R, S_{\tau_1}, \pub_{\tau_1}, \abs_{\tau_1}) \in \mathtt{Conf}_{\tau_2}\]
where
\begin{itemize}
\item $R$ is the initial repository obtained from client $C_{\tau}$;
\item $S_{\tau_1}$ has the same domain as the initial store $S$ obtained from client $C_\tau$ and defines values $\mathtt{cnt} \mapsto len(\tau_1) / 2$ and $\mathtt{ref}_i \mapsto m_i$ for all $m_i$ revealed in $\tau_1$;
\item $\pub_{\tau_1} = \abs \uplus \{ m_in \in \abs' \mid m_i \in \tau_1 \}$, for $\abs,\abs'$ as defined initially in $C_{\tau}$;
\item $\abs_{\tau_1} = \pub \uplus \{ m_in \in \pub' \mid m_i \in \tau_1 \}$, for $\pub,\pub'$ as defined initially in $C_{\tau}$;
\item and $\eval_{\tau_1} = f( {\lceil {\tau_1} \rceil} )$ where $\lceil \tau \rceil$ removes all closed calls in $\tau$ as defined in 
\[\lceil \tau \rceil = \begin{cases}
\lceil \tau' \tau''' \rceil & \text{if $\tau$ is of the form $\tau' \mathtt{call}(m,v) \tau'' \ret(m,v) \tau'''$}\\
\tau &\text{otherwise}
\end{cases}\]
and 
\begin{align*}
f&(\tau' \mathtt{call}(m,v)_o) = \\
&\quad(\letm{x = \bullet}{\mathtt{cnt}\!+\!\!+; \mathtt{meth} := 0; \mathtt{val}_{\theta_2} := x; \mathtt{oracle}()};!\mathtt{val}_{\theta_2},m) :: f(\tau')\\
f&(\mathtt{call}(m,v)_o) =\\ 
&\quad(\letm{x = \bullet}{\mathtt{cnt}\!+\!\!+; \mathtt{meth} := 0; \mathtt{val}_{\theta_2} := x; \mathtt{oracle}();!\mathtt{val}_{\theta_2}},m)::[]\\
f&(\tau' \mathtt{call}(m,v)_p) =  m :: f(\tau')\\
f&(\mathtt{call}(m,v)_p) = m::[]
\end{align*}
\end{itemize}
\end{definition}

\begin{lemma}\label{lem:clientconfigs}
Let library trace $\tau_L$ be of the form $\tau_1 \tau_2$, such that $\tau_1$ is a prefix of $\tau_L$. For all configurations $\mathbb{C}_{\tau_2} \in \mathtt{Conf}_{\tau_2}$, $\mathbb{C}_{\tau_2}$ produces $\tau_2$.
\end{lemma}
\begin{proof}
Let $\tau_L$ be a library trace of the form $\tau_p \tau$. We prove that $\mathbb{C}_{\tau}$ produces $\tau$ for all $\mathbb{C}_{\tau} \in \mathtt{Conf}_\tau$ by induction on the length of $\tau$.
\paragraph*{\boldemph{Base Cases:}}
\begin{itemize}
\item if $\tau = \mathtt{call}(m,v)$, then we know $\mathbb{C}_\tau \to (m::\eval_{\tau_p}, m v, \dots)_p$ produces a valid OQ move since $m$ must have been revealed as an initial public name or in $\tau_p$ for it to appear as a call at this point in the trace.
\item if $\tau = \mathtt{ret}(m,v)$, then we know $\mathbb{C}_\tau \to (\eval', v, \dots)_p$, where $\eval_{\tau_p} = \mathtt{call}(m,v')::\eval'$, produces a valid OA move since $m$ must appear at the top of the evaluation stack for a return to appear at this point in the trace.
\end{itemize}
We thus have base cases for \emph{odd length suffixes}.
\paragraph*{\boldemph{Inductive Cases:}}
\begin{itemize}
\item if $\tau = \mathtt{call}(m,v) \mathtt{call}(m',v') \tau'$, then we have the OQ move 
\[\mathbb{C}_\tau \to (m::\eval_{\tau_p}, m v, \dots)_p\twoheadrightarrow (m::\eval_{\tau_p},\mathtt{oracle}();!\mathtt{val}_{\theta_2}, \dots)_p \to (\dots, E[m' v'],\dots)_p\] 
where $E$ is $(E');!\mathtt{val}_{\theta_2}$ and $E'$ is defined from line 26 to line 28 in the client code, which correctly updates the store. So far, $\mathbb{C}_\tau$ produces the same trace up to the next move. We then have the PQ move
\[(m::\eval_{\tau_p},E[m' v'],\dots)_p \to ((E,m')::m::\eval_{\tau_p},\dots)_o\]
which produces the next valid move. At this point, we can observe that $((E,m')::m::\eval_{\tau_p},\dots)_o \in \mathtt{Conf}_\tau'$, so we know $\tau'$ is produced by the inductive hypothesis. Thus, $\tau$ is produced.

\item if $\tau = \mathtt{call}(m,v) \mathtt{ret}(m',v') \tau'$, since we have a return move as the second move this time, we have the OQ move
\[\mathbb{C}_\tau \to (m::\eval_{\tau_p}, m v, \dots)_p\twoheadrightarrow (m::\eval_{\tau_p},\mathtt{val}_{\theta_2} := v';!\mathtt{val}_{\theta_2}, \dots)_p \to (\dots,v',\dots)_p\]
which produces the first move. We then have the PA move
\[(\eval_{\tau_p},v',\dots)_p \to (\eval',\dots)_o\]
which produces the second move since $\eval_{\tau_p}$ must be of the form $m'::\eval'$. As before, since the store has been correctly updated by internal moves, $(\eval',\dots)_o \in \mathtt{Conf}_\tau'$, so we know $\tau'$ is produced by the inductive hypothesis. Thus, $\tau$ is produced.

\item if $\tau = \mathtt{ret}(m,v) \mathtt{call}(m',v') \tau'$, then it must be the case that $\eval_\tau = (\letm{x = \bullet}{\mathtt{cnt}\!+\!\!+; \mathtt{meth} := 0; \mathtt{val}_{\theta_2} := x; \mathtt{oracle}()},m)::\eval'$. We have the OA move
\[\mathbb{C}_\tau \to (\eval', \letm{x = v}{\dots}, \dots)_p\twoheadrightarrow (\eval',\mathtt{oracle}();!\mathtt{val}_{\theta_2}, \dots)_p \to (\eval',E[m' v'],\dots)_p\]
where $E$ is the context for $\mathtt{oracle}$, which produces the first move. From here we have OQ move
\[(\eval',E[m' v'],\dots)_p \to ((E,m')::\eval',\dots)_o\]
which produces the second move. Since the store is correctly updated internally, we know $((E,m')::\eval',\dots)_o \in \mathtt{Config}_\tau'$, so $\mathbb{C}_{\tau'}$ produces $\tau'$ by the inductive hypothesis. Thus, $\tau$ is produced. 

\item if $\tau = \mathtt{ret}(m,v) \mathtt{ret}(m',v') \tau'$, we have the OA move
\[\mathbb{C}_\tau \to (\eval', \letm{x = v}{\dots}, \dots)_p\twoheadrightarrow (\eval',!\mathtt{val}_{\theta_2}, \dots)_p \to (\eval',v',\dots)_p\]
which produces the first move. From here, we have PA move
\[(\eval',v',\dots)_p \to (\eval'',\dots)\]
since $\eval'$ must have been of the form $m'::\eval''$ for a return to $m'$ to appear on the trace.
Since the internal moves correctly update the store, we know that $(\eval'',\dots) \in \mathtt{Config}_\tau'$, so $\mathbb{C}_{\tau'}$ produces $\tau'$ by the inductive hypothesis. Thus $\tau$ is produced.

\end{itemize}
If $\tau'$ is empty, these serve as base cases for \emph{even length suffixes}. With all cases proven (odd and even base cases, and the inductive cases), we have that $\tau$ is always possible to produce with any $\mathbb{C}_\tau \in \mathtt{Conf}_\tau$.
\end{proof}

{\paragraph*{\boldemph{Theorem \ref{thm:definability} (Definability)}}
	Let $L$ be a library and $(\tau,\rho)\in\sem{L}$. There is a good client compatible with $L$ such that $(\tau,\rho')\in\sem{C}$ for some $\rho'$.}\\

\begin{proof}
Given a library $L$ and trace produced $\tau$, we construct client $C_{\tau}$. Since $C_\tau$ has a main method, we begin from a proponent configuration $(\mathtt{oracle}(),[],R,\abs,\pub)_p$. Since the library cannot return without being called first, we know the next move is a call, so $\tau$ is of the form $\mathtt{call}(m,v) \tau'$. Thus, we have the following transitions
\[([],\mathtt{oracle}(),R,\abs,\pub)_p \twoheadrightarrow ([],E[m v],R,\abs,\pub)_p \to ((E,m)::[],R,\abs',\pub)_o \]
From this point, if $\tau'$ is empty, we have shown that $\tau$ can be produced by $C_\tau$. If $\tau'$ is not empty, we have a trace $\tau$ with suffix $\tau'$ and prefix $\mathtt{call}(m,v)$. By Lemma~\ref{lem:clientconfigs}, we know that $\tau'$ can be produced by any configuration in $\mathtt{Config}_{\tau'}$. Since $((E,m)::[],R,\abs',\pub)_o \in \mathtt{Config}_{\tau'}$, we know that $((E,m)::[],R,\abs',\pub)_o$ is able to produce $\tau'$. We thus have that $C_\tau$ can produce $\tau$.
\end{proof}

\subsection{Extensional Equivalence of O-Refreshing Moves}

\paragraph*{\boldemph{Lemma \ref{lem:Orefresh} (O-Refreshing)}}
	Given a concrete configuration $\rho$, the following are equivalent:
	\begin{enumerate}
\item $\rho$ fails using any kinds of transitions
\item $\rho$ fails using only $O$-refreshing transitions
	\end{enumerate}

\begin{proof}
Let us consider two games starting from $\rho$: (A) is allowed to play any kind of moves, while (B) is only allowed to play $O$-refreshing moves. We thus want to show that (A) and (B) are both allowed to reach an assertion violation.

\paragraph*{\boldemph{(2)$\implies$(1):}} We know that (A) is allowed to play all the moves that (B) can play since (A) can play any moves, including $O$-refreshing moves. Thus, this direction holds.

\paragraph*{\boldemph{(1)$\implies$(2):}} Since we start from the same $\rho$ in (A) and (B), by Lemma~\ref{lem:phantom}, we know $\rho$ fails in (B) if it fails in (A). Given we know (A) fails by assumption, this direction holds.
\end{proof}

The above result requires the following lemma, which in turn requires some definitions. First, we call a name \emph{phantom} if it is an opponent name created by refreshing a proponent name through an $O$-refreshing transition that has some equivalent original name in the non-refreshing semantics. We assume a method to identify phantom names by keeping track of them with regard to the non-refreshing semantics as computation progresses. We thus say that a configuration $\rho$ that is reached through $O$-refreshing transitions has a corresponding phantom names dictionary $\Phi$ that maps all phantom names $m$ in $\rho$ to their proponent-owned original names $\hat{m}$ in $\Phi(\rho)$. Let us also define a set $\abs_\Phi \subseteq \abs$ for all the phantom names in $\abs$.

\begin{lemma}\label{lem:phantom}
Given a configuration $\rho$ with corresponding phantom names $\Phi$, it is the case that $\rho$ fails through $O$-refreshing transitions if $\Phi(\rho)$ fails.
\end{lemma}

\begin{proof}
Let (A) be a standard semantics where any moves are allowed. Let (B) be a semantics where only $O$-refreshing transitions are allowed. Suppose (B) starts from a configuration $\rho$ and has phantom names $\Phi$. We show this by induction on the number of steps to reach $\rho$. Let us consider proponent moves first, so $\rho = (\eval,M,R,S,\pub,\abs)_p$. Suppose $\Phi(\rho) \twoheadrightarrow{\tau} (\dots,\assert(0),\dots)$ in (A), by case analysis on $M$, we have the following.

\begin{enumerate}
\item $M$ is not of the form $E[mv]$ or is of the form $E[mv]$ where $m\in\pub$:

Let $\Phi(\rho) \to \hat\rho'$ via (A) semantics. Since $\rho$ is a proponent configuration, and the language features no name comparison, we know that the semantics are not affected by opponent names. Thus, we know $\hat\rho' = \Phi(\rho')$, so $\rho \to \rho'$ via (B). By the inductive hypothesis on $\hat\rho'$ and $\rho'$, we know (A) and (B) both fail.

\item $M$ is of the form $E[mv]$ and $m \in (\abs \setminus \abs_\Phi)$ ($m$ is not a phantom name):

Let $\Phi(\rho) \xrightarrow{\mathtt{call}(m,\hat v)} \hat\rho'$ in (A). It must be the case $\hat\rho' \xrightarrow{cr(\hat m',\hat v')} \hat \rho''$ for some call or return $cr$, since $\hat\rho'$ cannot fail without passing control to the proponent.

With (B), we know $\rho \xrightarrow{\mathtt{call}(m,v)} \rho' \xrightarrow{cr(m',v')} \rho''$. Extending $\Phi$, we get $\Phi' = \Phi[m_i' \mapsto \hat m_i']$ for every $m_i',\hat m_i \in v',\hat v'$. Thus, we have $\Phi'(\rho'') = \hat\rho''$. By the inductive hypothesis on $\rho''$, $\hat\rho''$ and $\Phi'$, we know (A) and (B) fail.

\item $M$ is of the form $E[mv]$ where $m \in \abs_\Phi$ ($m$ is a phantom name):

Let $\Phi(m) = \hat m$. We have two cases on $\hat m$:
\begin{enumerate}
\item If $\hat m \in \abs$, then we have the same situation as before.
\item If $\hat m \in \pub$, then we know $\hat\rho \to (\dots,\hat E[(R(\hat m)) \hat v],\dots)$ in (A). In (B), we have $\rho \xrightarrow{\mathtt{call}(m,v)} \rho'$. Since $\hat m$ must have been revealed to the opponent at some point in order for it to have been refreshed by (B), we have $\rho' \xrightarrow{\mathtt{call}(m,v)} (\dots,E[R(\hat m) v'],\dots)$. Extending $\Phi$ to account for the indirect call of $\hat m$, we have $\Phi' = \Phi[m_i \mapsto \hat m_i]$ for every $m_i \in v'$ and $\hat m_i \in \Phi(v)$. Thus, we have $\Phi'(\dots,E[R(\hat m) v'],\dots) = (\dots,\hat E[(R(\hat m)) \hat v],\dots)$, so by the inductive hypothesis on them, we know (B) fails.
\end{enumerate}
\end{enumerate}

For the opponent moves, the cases are captured for every move $\hat\rho \xrightarrow{cr(m,\hat v)} \hat\rho'$ in (A) and every move $\rho \xrightarrow{cr(m,v)} \rho'$ in (B) by extending $\Phi$ to be $\Phi' = \Phi[m_i \mapsto \hat m_i]$ for every name $m_i \in v$ and $\hat m_i \in \hat v$ introduced in the move. With this, by the inductive hypothesis on $\rho'$, $\hat \rho'$ and $\Phi'$, we know (B) fails in all the opponent cases. With this, we know (B) fails if (A) fails under $\Phi$.

\end{proof}

%%% Local Variables:
%%% mode: latex
%%% TeX-master: "paper"
%%% End:

}
\shorten{}{\section{Soundness of Symbolic Games}\label{apx:sound}
In this section we look into more detail into soundness of our symbolic semantics. %Intuitively, we want the generalisation of \emph{sound errors}, which is that any given symbolic configuration that reaches an error with a satisfiable path condition if and only if the equivalent concrete configuration reaches the same error by executing the moves provided in the counter example. %We formalise the argument in Theorem~\ref{sound_bisim}.
\\\\
{\textbf{Lemma~\ref{sound_bisim}}~~{\em
Let $\rho,\rho'$ be a concrete and symbolic configuration respectively, and let $\cal M$ be a model such that $\rho=_{\cal M}\rho'$. Then, $\rho\sim_{\cal M}\rho'$.
\begin{proof}
We want show that $\mathcal{R} = \{ (\rho,\mathcal{M},\rho_s) \mid \rho =_{\mathcal{M}} \rho_s \}$ is a bisimulation. First, we show that if $\rho \to \rho'$, being $O$-refreshing, then $\rho_s \to_s \rho_s'$ such that $(\rho',\mathcal{M}',\rho_s')$ is in $\mathcal{R}$ for some $\mathcal{M}' \supseteq \mathcal{M}$. By cases on the transition $\rho \to \rho'$:
\begin{enumerate}
\item If $\rho \to \rho'$ is one of the return moves, then we have the following possible transitions:
\begin{enumerate}
\item If $(\eval,E[\assert(0)],R,S,\pub,\abs,k)_p \not\to$, then we have the corresponding symbolic final configuration:
\[(\eval,E'[\assert(0)],R,\pub,\abs,\sigma,pc,k)_p\]
From the assumptions, we know that $\mathcal{M}\vDash pc \land \sigma^\circ$. It is also the case that $E'\{\mathcal{M}\}$ is equivalent to $E$, and $\rho'$ and $\rho_s'$ are equivalent terminal configurations.
\item If $(\emptyset,v,R,S,\pub,\abs,k)_p \not\to$, the proof is similar to (a).
\end{enumerate}
%And similarly for the opposite direction.
\item If $\rho \to \rho'$ is an (\textsc{Int}) move, we have that $\rho_s \to_s \rho_s'$ such that $\rho' \sim \rho_s'$ by soundness of the symbolic execution (Lemma~\ref{sound_symex}). %and similarly for the opposite direction.
\item If $\rho \to \rho'$ is a (\textsc{Pq}) move, then we have the following transition
\[(\eval,E[mv],R,S,\pub,\abs,k)_p \xrightarrow{\mathtt{call}(m,v)} ((m,E)::\eval,l_0,R',S,\pub',\abs,k)_o\]
with its corresponding symbolic equivalent
\[(\eval',E'[mv'],\dots,\sigma,pc,k)_p \xrightarrow{\mathtt{call}(m,v')}_s ((m,E')::\eval',l_0,\dots,\sigma,pc,k)_o\]
From the assumptions, we know $\mathcal{M}(v') = v$. In addition, since $E'[mv'] = E[mv]$ under $\mathcal{M}$, we have that $(m,E')::\eval' = (m,E)::\eval$, and similarly for other components, so $\rho' =_{\mathcal{M}} \rho_s'$, meaning $(\rho',\mathcal{M},\rho_s') \in \mathcal{R}$.
%For the other direction, we have an identical argument.
\item If $\rho \to \rho'$ is a (\textsc{Pa}) move, then we have the following transition
\[((m,l)::\eval,v,R,S,\pub,\abs,k)_p \xrightarrow{\mathtt{ret}(m,v)} (\eval,l,R',S,\pub',\abs,k)_o\]
with its corresponding symbolic equivalent
\[((m,l)::\eval',v',\dots,\sigma,pc,k)_p \xrightarrow{\mathtt{ret}(m,v')}_s (\eval',l,\dots,\sigma,pc,k)_o\]
From the assumptions, we know $\mathcal{M}(v') = v$. Since the original stacks are equivalent under $\mathcal{M}$, we have that $\eval =_{\mathcal{M}} \eval'$, and similarly for other components, so $\rho' =_{\mathcal{M}} \rho_s'$, meaning $(\rho',\mathcal{M},\rho_s') \in \mathcal{R}$. %For the other direction, we have an identical argument.
\item If $\rho \to \rho'$ is an (\textsc{Oq}) move, $O$-refreshing, then we have the following transition
\[(\eval,l,R,S,\pub,\abs,k)_o \xrightarrow{\mathtt{call}(m,v)} ((m,l+1)::\eval,mv,R,S,\pub,\abs',k)_p\]
with its corresponding symbolic equivalent
\[(\eval',l,\dots,\sigma,pc,k)_o \xrightarrow{\mathtt{call}(m,v')}_s ((m,l+1)::\eval',mv',\dots,\sigma,pc,k)_p\]
Let us choose $\mathcal{M}' = \mathcal{M}[v' \mapsto v]$. Since the original stacks are equivalent under $\mathcal{M}$, we have that $((m,l+1)::\eval) =_{\mathcal{M}} ((m,l+1)::\eval')$, and similarly for other components, so $\rho' =_{\mathcal{M}'} \rho_s'$, meaning $(\rho',\mathcal{M}',\rho_s') \in \mathcal{R}$. %For the other direction, we have an identical argument.
\item If $\rho \to \rho'$ is an (\textsc{Oa}) move, $O$-refreshing, then we have the following transition
\[((m,E)::\eval,l,R,S,\pub,\abs,k)_o \xrightarrow{\mathtt{ret}(m,v)} (\eval,E[v],R,S,\pub,\abs',k)_p\]
with its corresponding symbolic equivalent
\[((m,E')::\eval',l,\dots,\sigma,pc,k)_o \xrightarrow{\mathtt{ret}(m,v')}_s (\eval',E'[v'],\dots,\sigma,pc,k)_p\]
Let us choose $\mathcal{M}' = \mathcal{M}[v' \mapsto v]$. Since the original stacks are equivalent under $\mathcal{M}$, we have that $\eval =_{\mathcal{M}} \eval$. Additionally, since $\mathcal{M}'$ extends $\mathcal{M}$, we know that $E[v]=E'[v']$ under $\mathcal{M}'$, and similarly for the remaining components, so $\rho' =_{\mathcal{M}'} \rho_s'$, meaning $(\rho',\mathcal{M}',\rho_s') \in \mathcal{R}$.
%Additionally, we update $\Phi[m_\Phi \mapsto m]$ for every phantom name $m_\Phi$ where $R(m_\Phi) = \lambda x.m x$ in $v'$.
%For the other direction, we have an identical argument after replacing $m$ with $\Phi(m)$ on the symbolic transition.
\end{enumerate}
The opposite direction is treated with similarly.
\end{proof}}

{\lemma[Soundness of symbolic execution]{\label{sound_symex}}{
For any concrete configuration $\eta = (M,R,S,k)$ and symbolic configuration $\eta' = (M',R',\sigma,pc,k)$, given an assignment $\mathcal{M} \vDash pc \land \sigma^\circ$ such that $M =_\mathcal{M} M'$, it is the case that $\eta \sim \eta'$.
\begin{proof}
Let $\mathcal{R} = \{(\eta,\mathcal{M},\eta_s) \mid \eta =_\mathcal{M} \eta_s \}$ for any concrete configuration $\eta$ and symbolic configuration $\eta_s$. We want to show that $\mathcal{R}$ is a bisimulation. We now show that $\eta_s \to \eta_s'$ if $\eta \to \eta'$. By cases on $\eta \to \eta'$:
\begin{enumerate}
\item If we have a terminal rule, then we have the following cases.
\begin{enumerate}
\item for $(E[\assert(0)],R,S,k) \not\to$ we have the equivalent final configuration
\[(E'[\assert(0)],R',\sigma,pc,k)\]
Since $\eta =_\mathcal{M} \eta'$, and $\eta' =_\mathcal{M} \eta_s'$ since they are equivalent terminal configurations, it is the case that $(\eta',\mathcal{M},\eta_s')\in \mathcal{R}$.
\item for $(v,R,S,k) \not\to$ we have a similar proof to (a).
\end{enumerate}
%Similarly in the opposite direction.
\item If $(E[\assert(i)],R,S,k) \to (E[()],R,S,k)$ where $(i \neq 0)$, then we have the equivalent symbolic transition
\[(E'[\assert(i)],R',\sigma,pc,k) \to (E'[()],R',\sigma,pc,k)\]
By assumption, we know $E =_\mathcal{M} E'$ and $R =_\mathcal{M} R'$, and similarly for other components, so $\eta' =_\mathcal{M} \eta_s'$. As such, we know $(\eta',\mathcal{M},\eta_s')\in \mathcal{R}$. %Similarly in the opposite direction.
\item If $(E[!r],R,S,k) \to (E[S(r)],R,S,k)$, then we have the equivalent symbolic transition
\[(E'[!r],R',\sigma,pc,k) \to (E'[\sigma(r)],R',\sigma,pc,k)\]
Since $\eta =_\mathcal{M} \eta_s$, we know that $S =_\mathcal{M} \sigma$, meaning that $\sigma(r)\{\mathcal{M}\} = S(r)$. Thus, $(\eta',\mathcal{M},\eta_s') \in \mathcal{R}$. %Similarly in the opposite direction.
\item If $(E[r:=v],R,S,k) \to (E[()],R,S[r\mapsto v],k)$, then we have the equivalent symbolic transition
\[(E'[r:=v'],R',\sigma,pc,k) \to (E'[()],R',\sigma[r \mapsto \sigma(v')],pc,k)\]
Since $\eta =_\mathcal{M} \eta_s$, we know that $S =_\mathcal{M} \sigma$ and $v' =_\mathcal{M} v$, meaning that $\sigma[r\mapsto v']\{\mathcal{M}\} = S[r \mapsto v]$. Thus, $(\eta',\mathcal{M},\eta_s') \in \mathcal{R}$. %Similarly in the opposite direction.
\item If $(E[\pi_j\pair{v_1,v_2}],R,S,k) \to (E[v_j],R,S,k)$, then we have the equivalent symbolic transition
\[(E'[\pi_j\pair{v_1' ,v_2'}],R',\sigma,pc,k) \to (E'[v_j'],R',\sigma,pc,k)\]
Since $\eta =_\mathcal{M} \eta_s$, we know that $\pair{v_1,v_2} =_\mathcal{M} \pair{v_1',v_2'}$, so $v_j'\{\mathcal{M}\} = v_j$. Thus, $(\eta',\mathcal{M},\eta_s') \in \mathcal{R}$. %Similarly in the opposite direction.
\item If $(E[i_1 \oplus i_2],R,S,k) \to (E[i],R,S,k)$ where $i = i_1 \oplus i_2$, prove as above.
\item If $(E[\lambda x.M],R,S,k) \to (E[m],R[m \mapsto \lambda x.M],S,k)$, then we have the equivalent symbolic transition
\[(E'[\lambda x.M'],R',\sigma,pc,k) \to (E'[m],R'[m \mapsto \lambda x.M'],\sigma,pc,k)\]
Since $\eta =_\mathcal{M} \eta_s$, we know that $E[m] =_\mathcal{M} E[m']$, so $v_j'\{\mathcal{M}\} = v_j$. Additionally, we know $M = M'\{\mathcal{M}\}$, so $R'[m \mapsto \lambda x.M'] =_\mathcal{M} R[m \mapsto \lambda x.M]$. Thus, $(\eta',\mathcal{M},\eta_s') \in \mathcal{R}$. %Similarly in the opposite direction.
\item If $(E[\ifm{0}{M_1}{M_0}],R,S,k) \to (E[M_0],R,S,k)$, then we have the equivalent symbolic transition
\[(E'[\ifm{0}{M_1'}{M_0'}],R',\sigma,pc,k) \to (E'[M_0'],R',\sigma,pc,k)\]
Since $\eta =_\mathcal{M} \eta_s$, we know that $E[M_0] =_\mathcal{M} E[M_0']$. Thus, $(\eta',\mathcal{M},\eta_s') \in \mathcal{R}$. %Similarly in the opposite direction.
\item If $(E[\ifm{i}{M_1}{M_0}],R,S,k) \to (E[M_1],R,S,k)$ where $i \neq 0$, prove as above.
\item If $(E[\letm{x=v}{M}],R,S,k) \to (E[M\{v/x\}],R,S,k)$, then we have the equivalent symbolic transition
\[(E'[\letm{x=v'}{M'}],R',\sigma,pc,k) \to (E'[M'\{v'/x\}],R',\sigma,pc,k)\]
Since $\eta =_\mathcal{M} \eta_s$, we know that $E[M] =_\mathcal{M} E[M']$ and $v'\{\mathcal{M}\} = v$, so $E[M\{v/x\}] =_\mathcal{M} E[M'\{v'/x\}]$. Thus, $(\eta',\mathcal{M},\eta_s') \in \mathcal{R}$. %Similarly in the opposite direction.
\item If 
$\begin{aligned}[t]
&(E[\letrec{f=\lambda x.M'}{M}],R,S,k)\\ 
&\quad \to (E[M\{m/f\}],R[m \mapsto \lambda x.M'\{m/f\}],S,k)
\end{aligned}$ \\
prove by combining cases (7) and (10).
\item If $(E[mv],R,S,k) \to (E[\evalbox{M\{v/y\}}],R,S,k+1)$, prove like (10).
\item If $(E[\evalbox{v}],R,S,k) \to (E[v],R,S,k-1)$, then we have the equivalent symbolic transition
\[(E'[\evalbox{v'}],R',\sigma,pc,k) \to (E'[v'],R',\sigma,pc,k-1)\]
Since $v =_\mathcal{M} v'$, it is the case that $(\eta',\mathcal{M},\eta_s') \in \mathcal{R}$. %Similarly in the opposite direction.
\end{enumerate}
In the opposite direction, all cases are treated similarly to the ones above, but we now additionally have symbolic branching cases not directly covered by the previous cases.
\begin{enumerate}
\item If $(E[\assert(\kappa)],R,\sigma,pc,k) \to (E[\assert(0)],\sigma,pc \land (\sigma(\kappa) = 0))$, then there exists $\mathcal{M}$ such that $E[\assert(\kappa)]$ evaluates to $E[\assert(0)]$, which requires it to satisfy $(\sigma(\kappa) = 0)$. As such, we know $\mathcal{M} \vDash \sigma(\kappa) = 0$, meaning that $0 =_\mathcal{M} \kappa$. We thus have the following equivalent concrete configuration
\[(E'[\assert(0)],R',S,k) \not\to\]
which holds since $\eta'$ and $\eta_s'$ are equivalent terminal configurations.
\item If $(E[\assert(\kappa)],R,\sigma,pc,k) \to (E[()],\sigma,pc \land (\sigma(\kappa) \neq 0))$, prove as above.
\item If $(E[v_1 \oplus v_2],R,\sigma,pc,k) \to (E[\kappa],R,\sigma[\kappa \mapsto \sigma(v_1) \oplus \sigma(v_2)],pc,k)$, then we have the following equivalent concrete transition
\[(E'[i_1 \oplus i_2],R',S,k) \to (E'[i],R',S,k)\]
From the assumption, we know $i_1 \oplus i_2 =_\mathcal{M} \sigma(v_1) \oplus \sigma(v_2)$, so by choosing $\mathcal{M}' = \mathcal{M}[\kappa \mapsto i]$, we have that $\eta'$ and $\eta_s'$ are equivalent under $\mathcal{M}$. As such, this case holds.
\item If $(E[\ifm{\kappa}{M_1}{M_0}],R,\sigma,pc,k) \to (E[M_0],R,\sigma,pc \land (\sigma(\kappa)=0),k)$, then there must exist a model $\mathcal{M}\vDash \kappa = 0$. We thus have the following equivalent concrete transition
\[(E'[\ifm{0}{M_1'}{M_0'}],R',S,k) \to (E'[M_0'],R',S,k)\]
From the assumption, we know $M_0 =_\mathcal{M} M_0'$, so $\eta'$ and $\eta_s'$ are equivalent under $\mathcal{M}$. As such, this case holds.
\item If $(E[\ifm{\kappa}{M_1}{M_0}],R,\sigma,pc,k) \to (E[M_1],R,\sigma,pc \land (\sigma(\kappa)\neq 0),k)$, prove as above.
\end{enumerate}
\end{proof}
}}

%%% Local Variables:
%%% mode: latex
%%% TeX-master: "paper"
%%% End:

}
\shorten{}{\section{Correctness of call counters}\label{app:bounds}
%%%%%%%%%%%%%%%%%%%%%%%%%%%%
We prove our game semantics can be bounded, that is, games on independent components will always terminate if we bound the call counters. More precisely, Lemma~\ref{termination} states that our game semantics is strongly normalising when call counters are bounded, meaning that every transition sequence produced from a given configuration is finite. To do this, we will first define classes for ordering of moves. 

For any transition sequence $\rho_0 \to \dots \to \rho_i \to \dots$ and each $i>0$, we have the following two classes of configurations:
\begin{enumerate}[label=(\Alph*)]
\item either $\sizeof{\rho_{i}} < \sizeof{\rho_{i-1}}$, or
\item there exists $j < i-1$ such that $\sizeof{\rho_i} < \sizeof{\rho_j}$
\end{enumerate}
where $\sizeof{\rho} = (k_0 - k,\sizeof{M},l_0 - l)$ is the \emph{size} of $\rho$, and $\sizeof{\rho} < \sizeof{\rho'}$ is defined by the lexicographic ordering of the triple $(k_0 - k,\sizeof{M},l_0 - l)$, with bounds $k_0$ and $l_0$ such that $k \leq k_0$ and $l \leq l_0$ for semantic transitions to be applicable. If not present in the configuration, we look at the evaluation stack $\eval$ to find the top-most missing component. In other words, opponent configurations will have size $(k_0 - k,\sizeof{M},l_0 - l)$ where $E$ is the top-most one in $\eval$, whereas proponent configurations will have size $(k_0 - k,\sizeof{M},l_0 - l)$ where $l$ is the top-most one in $\eval$.

{\paragraph*{\boldemph{Theorem~\ref{termination}}}\em
For any concrete game configuration $\rho$ with bounds $k_0$ and $l_0$ for their corresponding counters $k$ and $l$, the semantics of $\rho$ is strongly normalising.}

\begin{proof}
We approach the proof two steps: (1) classify all possible transitions $\rho$ can make, thus classifying all reachable configurations, and (2) prove that the classes form a terminating sequence. For (1), considering all moves available to $\rho$, we have the following cases.
\begin{enumerate}
\item If $\rho \to \rho'$ is an (\textsc{Int}) move, we have two possibilities.
\begin{enumerate}
	\item For a transition $(E[\evalbox{v}],R,S,k) \to (E[v],R,S,k+1)$, where $k+1 \leq k_0$,
	we have a class (B) configuration since there must be a $(E[mv],R,S,k)$ such that $(E[mv],R,S,k) \to^* (E[v],R,S,k)$ which is lexicographically ordered since $\sizeof{v} < \sizeof{mv}$.
	\item Every other transition sequence is class (A) since they reduce the size of the term.
\end{enumerate}
\item If $\rho \to \rho'$ is a (\textsc{Pq}) move, we have that $\rho'$ is a class (A) configuration since $(k,\sizeof{E},l_0) < (k,\sizeof{E[mv]},l_0 - l)$ by lexicographic ordering.
\item If $\rho \to \rho'$ is an (\textsc{Oa}) move, we have a transition
\[((m,E)::\eval,l,\dots,k)_o\xrightarrow{ret(m,v)}(\eval,E[v],\dots,k)_p\]
which must be a result of the prior proponent question
\[(\eval,E[mv],\dots,k)_p\xrightarrow{call(m,v)}((m,E)::\eval,l_0,\dots,k)_o\]
where $\eval$ has an $l'$ on top. We thus have the following sequence
\[(\eval,E[mv],\dots,k)_p\to^*(\eval,E[v],\dots,k)_o\]
where $(k,\sizeof{E[v]},l) < (k,\sizeof{E[mv]},l')$, so $\rho'$ is a class (B) configuration.

\item If $\rho \to \rho'$ is an (\textsc{Oq}) move, we have the transition
\begin{align*}
(\eval,l, \dots,k)_o &\xrightarrow{call(m,v)} ((m,l+1)::\eval,mv,\dots,k)_p \\
&\to ((m,l+1)::\eval,\evalbox{M\{v/x\}},\dots,k+1)
\end{align*}
Ignoring the configuration in between, we take 
\[(\eval,l, R,S,\pub,\abs,k)_o \xrightarrow{call(m,v)} ((m,l+1)::\eval,\evalbox{M\{v/x\}},R,S,\pub,\abs,k+1)_p\] 
to be our new transition. We thus have that $\rho'$ is a class (A) configuration since $(k_0 - (k+1),\sizeof{\evalbox{M\{v/x\}}},l_0 - (l+1)) < (k_0 - k,\sizeof{E},l_0 - l)$ by lexicographic ordering.
\item If $\rho \to \rho'$ is a (\textsc{Pa}) move, we have the transition
\[((m,l)::\eval,v,\dots,k)_p \xrightarrow{ret(m,v)} (\eval,l, \dots,k)_o \]
which must be the result of a prior opponent question
\begin{align*}
(\eval,l + 1, \dots,k)_o&\xrightarrow{call(m,v)} ((m,l)::\eval,\evalbox{M\{v/x\}},\dots,k+1)_p\\ 
&\to^* ((m,l)::\eval,\evalbox{v},\dots,k+1)_p\\
&\to ((m,l)::\eval,v,\dots,k)_p\\
&\xrightarrow{ret(m,v)} (\eval,l,\dots,k)_o
\end{align*}
where $E'$ is the topmost evaluation context in $\eval$. We thus have that $(k_0 - k,E',l_0 - l)<(k_0 - k,E',l_0 - (l+1))$, so $\rho'$ is a class (B) configuration.
\end{enumerate}
Now, for part (2), let us assume there is an infinite sequence
\[\rho_0 \to \dots \to \rho_j \to \dots \to \rho_i \to \dots\]
Since all reachable configurations fall into either (A) or (B) class, we know that the sequence must comprise only (A) and (B) configurations. In this infinite sequence, we know that all sequences of (A) configurations are in descending size, so (A) sequences cannot be infinite. We also observe that (B) configurations are padded with (A) sequences. For instance, if $\rho_i$ is a (B) configuration, and $\rho_j$ is its matching configuration, there may have nested (B) configurations between $\rho_j$ and $\rho_i$, as well as (A) sequences padding these.

Additionally, these (B) configurations can only occur as a return to a call, so we know they only occur together with the introduction of evaluation boxes $\evalbox{\bullet}$. Since these brackets occur in pairs and are introduced in a nested fashion, we know $\eval$ can only contain evaluation contexts with well-bracketed evaluation boxes, meaning that there cannot be interleaved sequences of (B) configurations where their target configurations intersect. More specifically, the sequence 
\[\rho_0 \to \dots \to \rho_j \to \dots \to \rho_j' \to \dots \to \rho_i \to \dots \to \rho_i' \to \dots\]
where $\rho_i'$ matches $\rho_j'$ and $\rho_i$ matches $\rho_j$ is not possible.

Now, ignoring all (A) and nested (B) sequences, we are left with an infinite stream of top-level (B) sequences which are also in descending order. Since starting size is finite, we cannot have an infinite stream of (B) sequences. Thus, the assumption does not hold, so our semantics is strongly normalising.
\end{proof}

{\lemma[Call counters preserved after application]{\label{kbounds}}{Given the following sequences of game moves:
		\begin{align*}
		(1)&~(\eval,E[M],R,S,\pub,\abs,k)_p \twoheadrightarrow (\eval,E[v],R',S',\pub',\abs',k')_p\\
		(2)&~((m,E)::\eval,l,R,S,\pub,\abs,k)_o \twoheadrightarrow (\eval,E[v],R',S',\pub',\abs',k')_p
		\end{align*}
		where in both (1) and (2) we apply $\twoheadrightarrow$ until we reach the first occurrence of $\eval$ and $E[\evalbox{\bullet}]$ in the sequence of moves, and $\twoheadrightarrow$ is the reflexive transitive closure of game transitions ($\to$), it must be the case that $k=k'$ in both (1) and (2).}
	\begin{proof}
		Suppose we have the following transition sequences
		\begin{align*}
		(1)&~(\eval,E[M],R,S,\pub,\abs,k)_p\twoheadrightarrow (\eval,E[v],R',S',\pub',\abs',k')_p\\
		(2)&~((m,E)::\eval,l,R,S,\pub,\abs,k)_o\twoheadrightarrow (\eval,E[v],R',S',\pub',\abs',k')_p
		\end{align*}
		By induction on the length of the transition sequence (1) and mutually on the length of (2), we have the following cases, where we say $IH_p$ and $IH_o$ for the inductive hypotheses of (1) and (2) respectively:
		
		\paragraph*{\bf Base cases:}
		\begin{itemize}
			\item \textbf{Case (1):} If $M = v$, then $(\eval,E[v],R,S,\pub,\abs,k)_p$ is a zero-step transition. This case holds since $k=k$.
			\item \textbf{Case (2):} If the opponent returns, then we have a one-step transition
			\begin{align*}
			&((m,E)::\eval,l,R,S,\pub,\abs,k)_o\\
			&\quad \xrightarrow{\mathtt{ret}(m,v)}
			(\eval,E[v],R',S,\pub,\abs',k)_p
			\end{align*}
			This case holds since $k=k$.
		\end{itemize}
		
		\paragraph*{\bf Inductive cases (1):}
		\begin{itemize}
			\item if the sequence contains only internal moves, i.e. no call to the opponent is made, then we have the following transition sequence by the assumption in (1) that a value is reached.
			\begin{align*}
			(\eval,E[M],R,S,\pub,\abs,k)_p\twoheadrightarrow (\eval,E[v],R',S',\pub',\abs',k')_p
			\end{align*}
			By the inductive hypothesis $IH_p$, we know that $k=k'$.
			\item if the sequence of internal moves gets stuck, i.e. a call to the opponent is made, then we have the following transition sequence where $m \notin dom(R')$.
			\begin{align*}
			(\eval,E[M],R,S,\pub,\abs,k)_p
            &\twoheadrightarrow(\eval,E[E'[mv]],R',S',\pub',\abs',k')_p\\
			&\xrightarrow{\mathtt{call}(m,v)}((m,E[E'[\bullet]])::\eval,l,R',S',\pub'',\abs',k')_o
			\end{align*}
			where $\eval$ is of the form $(m,l)::\eval'$. By our assumption in (1) and (2), we know that the configuration must eventually lead to a value $v$. As such, the following transition must eventually occur.
			\begin{align*}
			&((m,E[E[\bullet]])::\eval,l,R',S',\pub'',\abs',k')_o\\
			&\quad\twoheadrightarrow (\eval,E[E[v]],R',S',\pub',\abs',k'')_p
			\end{align*}
			By the inductive hypothesis $IH_o$, we know that $k'=k''$. In addition, by our assumption that a value must be reached, it is the case that the following transition occurs.
			\begin{align*}
			&(\eval,E[E[v]],R',S',\pub',\abs',k'')_p\\
			&\quad \twoheadrightarrow (\eval,E[v],R'',S'',\pub'',\abs'',k''')_p
			\end{align*}
			By the inductive hypothesis $IH_p$, we know that $k=k'''$.
		\end{itemize}
		
		\paragraph*{\bf Inductive cases (2):}
		\begin{itemize}
			\item if a call to the proponent is made, then we have the following transition.
			\begin{align*}
			&(\eval',l,R,S,\pub,\abs,k)_o\\ 
			&\quad \xrightarrow{\mathtt{call}(m',v)}((m',l+1)::\eval',m'v,R',S,\pub,\abs',k)_p
			\end{align*}
			from the assumption that a value must be reached, we know that the following transition occurs.
			\begin{align*}
			&((m',l+1)::\eval',m'v,R',S,\pub,\abs',k)_p\\
			&\quad \twoheadrightarrow ((m',l+1)::\eval',v,R'',S',\pub',\abs'',k')_p
			\end{align*}
			From the inductive hypothesis $IH_p$, we know that $k=k'$.
		\end{itemize}		
\end{proof}}

%%% Local Variables:
%%% mode: latex
%%% TeX-master: "paper"
%%% End:

}

% \subsection{Time Taken (s)}
% \includegraphics[width=\linewidth]{results/times.png}

% \subsection{Counter Examples Found}
% \includegraphics[width=\linewidth]{results/bugs.png}

\end{document}